\documentclass[11pt, letterpaper]{scrartcl}
\usepackage{gentium}
\usepackage{amsmath,amsthm,amssymb,amsfonts}
\usepackage{mathtools}
\usepackage{mathrsfs}
\usepackage{multirow}
\usepackage[inline]{enumitem}
\usepackage{hyperref}
\usepackage{bm}
\usepackage[margin=1in]{geometry}
\usepackage{colortbl}
\usepackage{booktabs}
\usepackage[mathscr]{euscript}
\usepackage{microtype}

\usepackage{sectsty}
\usepackage{setspace}
\usepackage{caption}
\newtheorem{assumption}{Assumption}
\newtheorem{theorem}{Theorem}

\newtheorem{lemma}{Lemma}
\newtheorem{proposition}{Proposition}
\theoremstyle{remark}
\newtheorem{remark}{Remark}

\newtheorem{procedure}{Procedure}

\usepackage[round]{natbib}
 \bibpunct[, ]{(}{)}{,}{a}{}{,}%

\DeclareMathOperator{\Var}{Var}
\DeclareMathOperator*{\argmin}{arg\,min}
\DeclareMathOperator*{\argmax}{arg\,max}

\begin{document}
%%%%%%%%%%%%%%%%

\title{\LARGE Distributionally Robust Selection of the Best}

\author{
{\large Weiwei Fan\thanks{Advanced Institute of Business and School of Economics and Management, Tongji University, 200092 Shanghai, China} }
\and
{\large L. Jeff Hong\thanks{School of Management and School of Data Science, Fudan University, 200433 Shanghai, China} }
\and
{\large Xiaowei Zhang\thanks{Corresponding author. Department of Management Sciences, College of Business, City University of Hong Kong, Kowloong Tong, Hong Kong. Email: \href{mailto:xiaowei.w.zhang@cityu.edu.hk}{xiaowei.w.zhang@cityu.edu.hk}}}
}

\date{}

% \author{\textbf{(Authors' names blinded for peer review)}}

\maketitle

\begin{abstract}
\noindent
\textbf{Abstract.} Specifying a proper input distribution is often a challenging task in simulation modeling. In practice, there may be multiple plausible distributions that can fit the input data reasonably well, especially when the data volume is not large. In this paper, we consider the problem of selecting the best from a finite set of simulated alternatives, in the presence of such input uncertainty. We model such uncertainty by an ambiguity set consisting of a finite number of plausible input distributions, and aim to select the alternative with the best worst-case mean performance over the ambiguity set. We refer to this problem as robust selection of the best (RSB). To solve the RSB problem, we develop a two-stage selection procedure and a sequential selection procedure; we then prove that both procedures can achieve at least a user-specified probability of correct selection under mild conditions. Extensive numerical experiments are conducted to investigate the computational efficiency of the two procedures. Finally, we apply the RSB approach to study a queueing system's staffing problem using synthetic data and an appointment-scheduling problem using real data from a large hospital in China. We find that the RSB approach can generate decisions significantly better than other widely used approaches.

\end{abstract}

\section{Introduction}

Simulation is widely used to facilitate decision-making for stochastic systems. In general, the performance of a stochastic system depends on \textit{design} variables and \textit{environmental} variables. The former is controllable by the decision-maker, while the latter is not. By simulating the environmental variables, the decision-maker can estimate the system's mean performance for arbitrary values of the design variables. A crucial step for building a credible simulation model is to characterize the environmental variables with an appropriate probability distribution, typically referred to as the \textit{input distribution} in simulation literature. This is often difficult, mainly because of lack of enough data, or measurement error in the data, either of which causes uncertainty concerning the input distribution, i.e., \textit{input uncertainty}.

Input uncertainty has drawn substantial interest from the simulation community in the past two decades; see \cite{henderson2003input} for a survey. The existing work usually assumes that the input distribution belongs to a particular parametric family, but the parameters of the distribution need to be estimated. This assumption reduces the input uncertainty to the so-called \textit{parameter uncertainty} and the primary objective becomes to characterize the randomness of the simulation output that is amplified by the parameter uncertainty. For instance, \cite{ChengHolland97} use the delta method to approximate the variance of the simulation output and \cite{BartonSchruben01} use the bootstrap method.

However, in practice it is non-trivial to determine the proper parametric family. Indeed, there may be several plausible parametric families that fit the input data reasonably well if the data volume is not large.  For instance, in Section \ref{sec:scheduling}, we study an appointment-scheduling problem in a large hospital in China. The maximum number of operations of a particular type performed by a surgeon in the hospital in 2014 is 138, and goodness-of-fit tests reject neither the gamma distribution nor the lognormal distribution when fitting the data for the duration of operations. Notice that these two parametric families may result in qualitatively different performances of a stochastic system. For instance, a queueing system's behavior depends critically on whether its service times are light-tailed or heavy-tailed \citep{Asmussen03}. Therefore, in this paper we focus on the uncertainty in specifying the parametric family of the input distribution, instead of considering parameter uncertainty.

One approach to address this difficulty is Bayesian model averaging \citep{chick2001input}. It measures the stochastic system by the weighted average of its mean performance under different plausible input distributions, where the weights are specified by prior estimation of the likelihood that a particular plausible distribution is the ``true'' distribution. This approach takes an ``ambiguity-neutral'' viewpoint concerning the input uncertainty.

In this paper, we take a robust approach that adopts an ``ambiguity-averse'' \citep{Epstein99} viewpoint and uses the worst-case mean performance of all the plausible distributions to assess a stochastic system. Using the worst-case analysis to account for uncertainty has a long history in economic theory. \cite{Ellsberg61} argues that in a situation where probability distributions cannot be specified completely, considering the worst of all the plausible distributions might appeal to a conservative person. \cite{GilboaSchmeidler89} rationalize the ambiguity aversion by showing that an individual who considers multiple prior probability distributions and maximizes the minimum expected utility over these distributions would act in this conservative manner. However, we do not argue or suggest that worst-case analysis is better than the ``model-averaging'' approach. Instead, we believe that they are equally important and that decision-makers should consider different perspectives in order to be fully aware of the potential risks of a decision.

We focus on an important class of simulation-based decision-making problems. We assume that the design variables of the stochastic system of interest take values from a finite set, each of which is referred to as an alternative. The mean performance of an alternative is estimated via simulation and we are interested in selecting the ``best'' alternative. This is known as the selection of the best (SB) problem in simulation literature. Due to statistical noise inherent in the simulation procedure, the probability that the best alternative is not selected is nonzero regardless of the computational budget. Thus, the objective is to develop a selection procedure that selects the best alternative with some statistical guarantee; see \cite{kim2006selecting} for an overview. In this paper, we consider the SB problem in the presence of input uncertainty and solve it in a way that is robust with respect to input uncertainty.

\subsection{Main Contributions}

First, we model the input uncertainty as an ambiguity set consisting of finitely many plausible distribution families whose associated parameters are properly chosen. We then transform the SB problem in the presence of input uncertainty into a \textit{robust selection of the best} (RSB) problem. Each alternative has a distinctive mean performance for each input distribution in the ambiguity set, and its worst-case mean performance is used as a measure of that alternative. The best alternative is defined as the one having the best worst-case mean performance.

Second, assuming the ambiguity set is given and fixed, we propose a new indifference-zone (IZ) formulation and design two selection procedures accordingly. The IZ formulation was proposed by \cite{bechhofer1954single}. However, to cope with our robust treatment of input uncertainty, we redefine the IZ parameter, denoted by $\delta$, as the smallest difference between the \textit{worst-case} mean performance of two alternatives that a decision-maker considers worth detecting. Then, the statistical evidence for designing a proper selection procedure can be expressed as the probability of selecting an alternative that is within $\delta$ of the best alternative in terms of their worst-case mean performance. We develop a two-stage procedure and a sequential procedure with statistical validity, i.e., they guarantee achieving a probability of correct selection (PCS) that is no less than a pre-specified level in a finite-sample regime and an asymptotic regime, respectively.

Third, we extend standard numerical tests for the SB problem to the new setting and demonstrate the computational efficiency of the two proposed RSB procedures. In particular, the sequential RSB procedure's efficiency in terms of the required total sample size is insensitive to the IZ parameter $\delta$ when $\delta$ is small enough. This is appealing to a practitioner, because it enables $\delta$ to be set as small as possible so that the unique best alternative can be selected instead of some ``near-best'' one without worrying computational burden. Besides, the proposed sequential RSB procedure is carefully designed so that it requires a much smaller total sample size than a plain-vanilla sequential RSB procedure as the problem scale increases.

Fourth, we assess the RSB approach in a queueing simulation environment where the input data, and thus the ambiguity set, is subject to random variation. Specifically, we consider a multi-server queue with abandonment, whose service time has an unknown distribution. The decision of interest is the staffing level, i.e., the number of servers. The cost of the queueing system depends on waiting and abandonment of the customers as well as the staffing level. We compare the RSB approach with a common approach for input modeling in practice, that is, the decision-maker fits a group of distribution families to the input data and uses the ``best'' fitted one as if it were the true distribution. An extensive numerical investigation reveals that the RSB approach can generate a staffing decision that has a significantly lower and more stable cost. %This demonstrates usefulness of the RSB approach in the presence of input uncertainty.  

Finally, we apply the RSB approach to an appointment-scheduling problem using real data from a large hospital in China. We show that in the presence of deep input uncertainty, the scheduling decision generated by the RSB approach incurs significantly lower operating costs than other widely used approaches, including a so-called ``increasing order of variance'' scheduling rule, one that is commonly viewed as a good heuristic in healthcare practice and was theoretically shown to be the optimal scheduling rule under some robust framework \citep{Mak2015appointment}.

\subsection{Related Literature}

This paper is related to three streams of literature, i.e., simulation input uncertainty, robust optimization, and selection of the best. Studies of input uncertainty in simulation literature have focused on the impact of input uncertainty on simulation output analysis; for instance, constructing confidence intervals to reflect input uncertainty. A preferred approach is resampling, consisting of macro-replications, in each of which the input data is first resampled to construct an empirical distribution as the input distribution. The sampled empirical distribution is then used to drive the simulation model to estimate the performance of the involved stochastic system. Finally, the performance estimate is collected as a bootstrap statistic from each macro-replication and a dynamic confidence interval is constructed for the performance measure of interest. Representative articles include \cite{ChengHolland97} and \cite{BartonSchruben01}. Bayesian model averaging also relies on macro-replication, but each macro-replication begins with sampling from the posterior (based on the input data) of the plausible input distributions and then uses the sampled input distribution to drive the simulation model; see \cite{chick2001input}. Recently, \cite{BartonNelsonXie14} and \cite{XieNelsonBarton14}  have both studied the propagation of input uncertainty to the estimated performance of a stochastic system,
% through the stochastic kriging metamodel \citep{AnkenmanNelsonStaum10}, 
 using nonparametric bootstrapping and Bayesian analysis, respectively.

The above research essentially takes an ambiguity-neutral attitude to the input uncertainty rather than ambiguity averse attitude as we do. In addition, it concentrates on the performance analysis of a stochastic system for a fixed value of its design variables. Unlike our paper, they do not include optimizing the performance over the design variables. Optimizing performance in the presence of distributional uncertainty is a theme of robust optimization; see \cite{ben2009robust} for an introduction to this broad area. However, robust optimization literature generally does not consider cases in which an objective function is embedded in a black-box simulation model and can only be evaluated using random samples; an exception is \cite{hu2012robust} but they focus on parameter uncertainty of the input distribution.

There is also a vast literature regarding the SB problem. Selection procedures can be categorized into frequentist procedures or Bayesian procedures depending on the viewpoint adopted for interpreting the unknown mean performance of an alternative. The former treats it as a constant and can be estimated through repeated sampling. Representative frequentist selection procedures include \cite{rinott1978two,KimNelson01,ChickWu05,Frazier14}, and \cite{izfree2016}. All these adopt an IZ formulation and use PCS as a selection criterion, except \cite{ChickWu05} in which the selection criterion is set to be expected opportunity cost, and \cite{izfree2016} in which an IZ-free formulation that can save users from the burden of specifying an appropriate IZ parameter is proposed. The present paper follows a frequentist viewpoint as well. Bayesian procedures, on the other hand, view the unknown mean of an alternative as a posterior distribution conditionally on samples calculated by Bayes' rule. The main approaches used in the Bayesian framework include (i) optimal computing budget allocation \citep{HeChickChen07}, (ii) knowledge gradient \citep{FrazierPowellDayanik09}, (iii) expected value of information \citep{ChickBrankeSchmit10}, and (iv) economics of selection procedures \citep{ChickFrazier12}.

Few prior papers in SB literature address input uncertainty, except \cite{CorluBiller13,CorluBiller15}, which focuses on the subset-selection formulation instead of the IZ formulation, and \cite{SongNelsonHong15}, which finds that in the presence of input uncertainty, IZ selection procedures designed for the SB problems may fail to deliver a valid statistical guarantee of correct selection for some configurations of the competing alternatives. Unlike our paper, these three papers all take an ambiguity-neutral viewpoint.

The rest of paper is organized as follows. Section \ref{sec:robustframework} formulates the RSB problem. Sections \ref{sec:robsut_twostage} and \ref{sec:robust_sequential} develop the two-stage and sequential RSB procedures, respectively, and show their statistical validity. Section \ref{sec:robust_numerical} presents numerical experiments to demonstrate the computational efficiency of the two RSB procedures. In Section \ref{sec:queueing}, we verify statistical validity of the proposed RSB procedures and demonstrate usefulness of the RSB approach in the context of queueing simulation, a more realistic setting than that of Section \ref{sec:robust_numerical}. In Section \ref{sec:scheduling}, we apply the RSB approach to address an appointment-scheduling problem using real data from a large hospital in China. We conclude in Section \ref{sec:conclusion} and collect all the proofs and additional numerical results in Appendix.

\section{Robust Selection of the Best}\label{sec:robustframework}

Suppose that a decision-maker needs to decide among $k$ competing alternatives, i.e., $\mathcal{S}=\{s_1,s_2,\ldots,s_k\}$. For each $s_i$, $i=1,\ldots,k$, let $g(s_i,\xi)$ denote its performance given an input variable $\xi$. In practice, $\xi$ is typically random and follows probability distribution $P_0$. Notice that $P_0$ may differ between the alternatives, but we suppress its dependence on $s_i$ for the purpose of notational simplicity. Each alternative is then measured by its mean performance $\mathbb{E}_{P_0}[g(s_i,\xi)]$, $i=1,\ldots,k$. The decision-maker aims to select the best alternative from $\mathcal S$, which is defined as the one having the smallest mean performance,
\begin{eqnarray*}
\min_{s\in \mathcal{S}} \mathbb{E}_{P_0}[g(s,\xi)].
\end{eqnarray*}
This is known as the SB problem, and a great variety of selection procedures have been developed, aiming to provide a desirable statistical guarantee on the probability of selecting the best.

%\subsection{Ambiguity Set}
To date, the SB problem has been studied primarily under the premise that the distribution $P_0$ is known and fixed. However, this is hardly the case in real-world applications. We assume that the distribution $P_0$ belongs to an {\em ambiguity set} $\mathcal{P}$ that consists of a finite number of plausible distributions, i.e., $\mathcal{P}=\{P_1,P_2,\ldots,P_m\}$. The form of $\mathcal{P}$ is determined by the following common scenario in input modeling: modern simulation software, e.g., Input Analyzer of Arena \citep{KeltonSadowskiSwets09}, typically has a built-in functionality to fit input data to a specified parametric distribution family and to perform some goodness-of-fit tests (e.g., Kolmogorov-Smirnov test and chi-squared test). A preliminary exploration of the input data may suggest a set of plausible distribution families and they are then examined by the software one at a time. Hence, a typical example of $\mathcal{P}$ is such that each  $P_j$ belongs to a distinctive parametric family, whose parameters are estimated from the data and which is not rejected by the goodness-of-fit tests. Notice that the ambiguity set $\mathcal{P}$ constructed in this way will converge to the true input distribution as the data volume increases, provided that the true distribution family is included in the set of plausible distribution families.

Given the ambiguity set $\mathcal P$, we measure an alternative by its worst-case mean performance over $\mathcal P$ and denote the best alternative as the alternative with the smallest worst-case mean performance. Then, the SB problem in the presence of input uncertainty is formulated as
\begin{equation}\label{eq:minimax}
\min_{s\in \mathcal{S}} \max_{P\in \mathcal{P}}\mathbb{E}_P[g(s,\xi)],
\end{equation}
which we call the RSB problem. Our goal is to develop selection procedures that, upon termination, select the best alternative with a probability of at least a user-specified value $1-\alpha$, $(0<\alpha<1)$.

\begin{remark}
The formulation \eqref{eq:minimax} assumes that $\mathcal P$ is given and fixed. On its own, it does not address the issue of \emph{statistical consistency} in the sense that $\mathcal P$ converges to the unit set that contains only the true distribution $P_0$ as the size of the input data grows to infinity. Thus, this issue is not addressed by the RSB methodology developed here. For our methodology to perform correctly, certain mechanism needs to be implemented to ensure that all plausible distributions in $\mathcal P$ that are not $P_0$ would be discarded eventually as more input data becomes available. Using a goodness-of-fit test is one possible approach. But further theoretical work on this issue would be of interest. 
\end{remark}

\begin{remark}
There is a subtle but critical difference between the conventional SB context and the context of the present paper with regard to the concept of ``random sample''. In the former context, $P_0$ is known and the mean performance of each $s_i$ is estimated by a random sample of size $N$ of the simulation output $g(s_i, \xi)$ with $\xi$ generated from $P_0$, so the estimate depends on $N$. In the RSB context, however, $P_0$ is unknown and the distributions in  $\mathcal P$ all try to estimate $P_0$ based on a sample of it of size $\ell$ (i.e., the input data), so each $P_j\in \mathcal P$ depends on $\ell$. Therefore, in a RSB procedure the estimate of each alternative's mean performance under each $P_j$ generally depends on both $N$ and $\ell$. By assuming $\mathcal P$ is given and fixed, we essentially ignore the dependence on $\ell$. A more complete treatment  would account for the fact that each $P_j\in\mathcal P$ estimated from the input data is actually random and make $\ell$ a possible factor for designing a RSB procedure. But this is beyond the scope of the present paper.
\end{remark}

Before moving to next section, we first introduce some necessary notations and assumptions. Let ``system $(i,j)$'' represent the pair of decision $s_i$ and probability scenario $P_j$, and $g(s_i,\xi)$ with $\xi$ following distribution $P_j$ denote the random observation from system $(i,j)$; further, let $\mu_{ij} =\mathbb{E}_{P_j}[g(s_i,\xi)]$ and $\sigma_{ij}^2=\Var_{P_j}[g(s_i,\xi)]$. The following assumptions are imposed throughout the paper.

\begin{assumption}\label{asp:basic}
For each $i=1,2,\dots,k$, $\mu_{i1}\geq \mu_{i2}\geq\dots \geq\mu_{im}$. Moreover, $\mu_{11} < \mu_{21}\leq \ldots \leq\mu_{k1}$.
\end{assumption}

\begin{assumption}\label{asp:variance}
For each $i=1,2,\dots,k$ and $j=1,2,\dots,m$, $\sigma_{ij}^2<\infty$.
\end{assumption}

In the RSB problem \eqref{eq:minimax}, our objective is to identify for each alternative its corresponding worst-case probability scenario, which is irrelevant to how the probability scenarios are ordered in the ambiguity set. Without loss of generality, we allow the means $\mu_{ij}$'s to be of certain configuration presented in Assumption \ref{asp:basic}; otherwise, we can sort the means and relabel the systems in the desired order. Clearly, under Assumption \ref{asp:basic}, system $(i,1)$ yields the worst-case mean performance of alternative $i$, $i=1,\ldots,k$, and alternative 1 is the unique best alternative in \eqref{eq:minimax}.

Assumption \ref{asp:variance} states that for each $s_i$, the random performance $g(s_i,\xi)$ is of finite variance under each $P_j$ included in the ambiguity set. Considering $\Var_{P_0}[g(s_i,\xi)]<\infty$ in many practical situations, it is reasonable to choose the probability scenario $P$ yielding the finite $\Var_P[g(s_i,\xi)]$ for all $i$ as a candidate representative for $P_0$ and then include it into the ambiguity set. Besides, assuming a finite variance of the random performance of each system is common in SB literature.

\subsection{Indifference-Zone Formulation}
We adopt the IZ formulation to design RSB procedures. Under the IZ formulation, the sought procedures are expected to provide a lower bound for both the probability of correct selection (CS) and the probability of good selection (GS); see, e.g., \cite{ni2017efficient} for their definitions in the SB setting. Since the RSB problem is of a minimax structure different from SB problems, the CS and GS events need to first be carefully redefined.

Let $\delta$ be a pre-specified IZ parameter which is the smallest difference that the decision-maker deems worth detecting. If $\mu_{21}-\mu_{11}>\delta$, alternative 1 is better than the others by at least $\delta$, measured by their worst-case mean performance over $\mathcal P$, due to Assumption \ref{asp:basic}. We define the CS event as the event where alternative 1 is selected. If $\mu_{21}-\mu_{11}\leq\delta$, some ``good" alternatives exist, and their worst-case mean performances are within $\delta$ of alternative 1; decision-makers feel indifferent between those good alternatives and alternative 1. We define the GS event as the event where one of the good alternatives is selected. Hence, selecting alternative $i$ is a good selection if $\mu_{i1} - \mu_{11}\leq \delta$.

Subtlety exists in the definitions of CS and GS, and it is worth some remarks. Take CS for an example. In the presence of the ambiguity set $\mathcal P$, it may be tempting to define CS as selecting system $(1,1)$, which refers to a pair of the best alternative and its corresponding worst-case probability scenario. However, what matters to a decision-maker is to select the best alternative rather than identifying which input distribution yields the worst-case mean performance of the alternatives. This is because the selected alternative will be implemented later and the ambiguity set is merely used to evaluate the alternatives. 

In SB literature, most IZ procedures are designed for the situation when $\mu_{21}-\mu_{11}>\delta$, and thus it is conventional to say a procedure is statistically valid if the achieved probability of correct selection (PCS) is no smaller than a pre-specified value $1-\alpha$. Borrowing the notation from SB literature, we use PCS to denote a measure of statistical validity, but in the extended way. Particularly, we define PCS as the probability of CS if $\mu_{21}-\mu_{11}>\delta$ and the probability of GS otherwise. Then, this paper seeks RSB procedures with \textit{statistical validity} in the following form: given a pre-specified $\alpha\in(0,1)$
\begin{eqnarray}\label{eqn:robust}
\mathbb{P}\{\mu_{i^* 1}-\mu_{11}\leq \delta\}\geq 1-\alpha,
\end{eqnarray}
where $i^*$ is the index of the selected alternative upon termination of a procedure.

\subsection{Double-Layer Structure}

In view of the minimax formulation of the RSB problems \eqref{eq:minimax}, we propose a double-layer structure for designing RSB procedures. An inner-layer procedure aims to select system $(i,1)$, which produces the worst-case mean performance for alternative $i$ with at least a pre-chosen inner-layer PCS, for each $i=1,2,\ldots,k$. An outer-layer procedure, on the other hand, aims to select system $(1,1)$ from the inner-layer selected systems with at least pre-chosen outer-layer PCS. After the double-layer selection process, we expect system $(1,1)$ to be selected with at least the probability of $1-\alpha$. To that end, the PCS in each layer must be judiciously chosen so that the overall PCS is no less than $1-\alpha$ as in \eqref{eqn:robust}. The detailed discussion is deferred to Section \ref{sec:twostage_errorallocation}. Figure \ref{fig:GraphicRepresentation} illustrates the double-layer structure of RSB procedures.

\begin{figure}[ht]
\begin{center}
\caption{Two-layer Structure of RSB Procedures}\label{fig:GraphicRepresentation}
\includegraphics[width=0.8\textwidth]{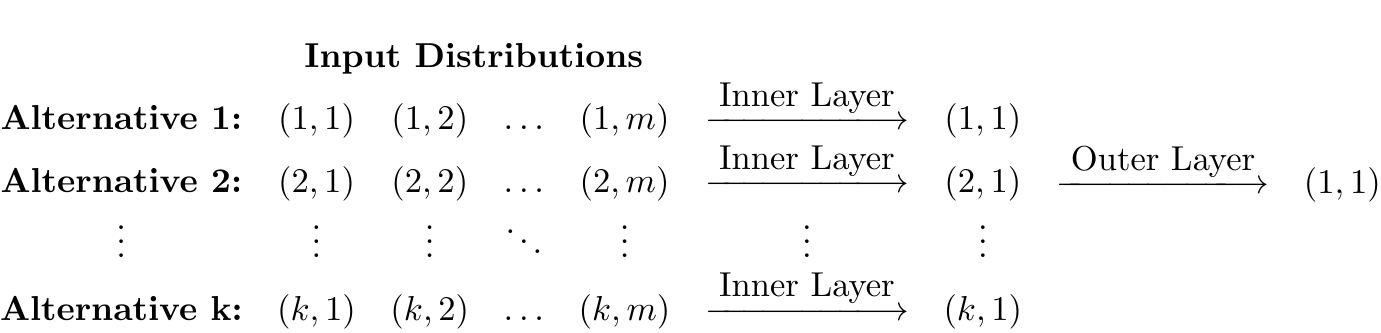}
\end{center}
\small{\textit{Note.}  Systems in the same column do not necessarily have the same input distribution. They are ordered to satisfy Assumption \ref{asp:basic}. The methodology developed in this paper even allows the alternatives to have different ambiguity sets.}
\end{figure}

\section{Two-Stage RSB Procedure}\label{sec:robsut_twostage}

In this section, we develop a two-stage RSB procedure with the statistical validity in the form of \eqref{eqn:robust}. In the first stage, we take a small number of samples from each system to calculate the sample variance of each pairwise difference. Based on them, we calculate the sample sizes needed for the second stage, and select the best alternative based on the sample means obtained after the second-stage sampling. 

The required sample size here resembles formally that of a typical two-stage procedure for the SB problem. However, the IZ parameter needs a special treatment to accommodate the double-layer structure. We wish to obtain an IZ for the RSB problem that has a given IZ parameter value. We do this by using two separate IZ parameter values, one corresponding to the inner-layer and the other to the outer-layer selection process, but calculated in a combined way so that the required IZ parameter value is obtained for the overall RSB problem. 

The procedure is based on pairwise comparisons between the systems, each of  which has a nonzero probability to yield an incorrect result due to  simulation noise. We need to specify how much probability of making an error in each pairwise comparison  is allowed, hereafter referred to as the \textit{error allowance} and denoted by $\beta$, in order to achieve the overall PCS. It is nontrivial to specify error allowances due to the double-layer structure.

Assuming that the inner-layer and outer-layer IZ parameters, denoted by $\delta_I$ and $\delta_O$, respectively, as well as the error allowance are known, we now present the two-stage RSB procedure (Procedure T). Specification of  $(\delta_I, \delta_O)$ and $\beta$ will be addressed in Section \ref{sec:robust_IZ_twostage} and Section \ref{sec:twostage_errorallocation}, respectively.

\vspace{5pt}

\begin{procedure}[Procedure T: Two-Stage RSB Procedure]\label{Proc:robust_twostage}

\begin{enumerate}
\item[]
\setcounter{enumi}{-1}
\item \textit{Setup.}
%Specify the IZ parameter $\delta$ and the PCS guarantee $1-\alpha$. 
Specify the inner-layer and outer-layer IZ parameters $(\delta_I, \delta_O)$,  the error allowance $\beta$, and the first-stage sample size $n_0\geq 2$. Set $h = t_{1-\beta, n_0-1}$ to be the $100(1-\beta)\%$ quantile of the Student's $t$ distribution with $n_0-1$ degrees of freedom.
\item \textit{First-stage sampling.}
Take $n_0$ independent replications $X_{ij,1},\ldots, X_{ij,n_0}$ of each system $(i,j)$. Compute the first-stage sample mean of each system and the sample variance of the difference between each pair of systems as follows
\[
\begin{array}{c}
\displaystyle\bar{X}_{ij}(n_0)=\frac{1}{n_0}\sum_{r=1}^{n_0} X_{ij,r}, \quad 1\leq i\leq k, \;1\leq j\leq m,\\
\displaystyle S_{ij,i'j'}^2 = \frac{1}{n_0-1}\sum_{r=1}^{n_0}\left[X_{ij,r}-X_{i'j',r}-(\bar{X}_{ij}(n_0)-\bar{X}_{i'j'}(n_0))\right]^2,  \quad 1\leq i,i'\leq k, \;1\leq j,j'\leq m.
\end{array}\]
\item \label{step:sampling_twostage}\textit{Second-stage sampling.}
Compute the total sample size $\displaystyle N = \max_{(i,j),(i',j')}\left\{n_0, N_{ij,i'j'}\right\}$, where
\[
N_{ij,i'j'} = \max \left\{\left\lceil \frac{h^2S_{ij,i'j'}^2}{\delta_I^2} \right\rceil,  \left\lceil \frac{h^2S_{ij,i'j'}^2}{\delta_O^2} \right\rceil\right\}, \quad 1\leq i,i'\leq k, \;1\leq j,j'\leq m, \]
with $\lceil x \rceil$ denoting the smallest integer no less than $x$. Take $N-n_0$ additional independent replications of each system.
\item \textit{Selection.}
Compute the overall sample mean of each system based on the $N$ replications
\[
\bar{X}_{ij}(N) = \frac{1}{N}\sum_{r=1}^{N}X_{ij,r}, \quad 1\leq i\leq k, \;1\leq j\leq m.
\]
Return $i^* =\displaystyle\argmin_{1\leq i\leq k} \max_{1\leq j\leq m}\bar{X}_{ij}(N)$ as  the best alternative. \hfil$\Box$

\end{enumerate}
\end{procedure}

\vspace{11pt}

Similarly to two-stage procedures for the SB problem, the two-stage RSB procedure is easy to implement and variance reduction techniques, such as common random numbers, can readily be applied to increase the efficiency of the algorithm, where the efficiency is defined in terms of the total sample size required.

\subsection{Inner-Layer and Outer-Layer IZ Parameters}\label{sec:robust_IZ_twostage}

The inner layer is used to estimate the worst-case mean performance of each alternative, while the outer layer is to compare the alternatives based on the estimates in the inner layer. If system $(i, j_i^*)$ is correctly selected in the inner layer to represent the worst system of alternative $i$, then by the definition of $\delta_I$ and Assumption \ref{asp:basic},
\begin{equation}\label{eqn:InnerIZ}
0<\mu_{i1}-\mu_{ij_i^*}\leq\delta_I, \quad i=1,\ldots,k.
\end{equation}

Suppose that alternative $i$ is $\delta$ away from alternative 1, i.e., $\mu_{i1}-\mu_{11}>\delta$. Then, alternative $i$ being eliminated by alternative 1 is necessary for the CS event. To that end, we need to ensure $\mu_{ij_i^*}-\mu_{1j_1^*} > \delta_O$,
when comparing $\mu_{1j_1^*}$ and $\mu_{ij_i^*}$ in the outer-layer. Notice that
\begin{equation*}\label{eqn:OuterIZ}
\mu_{ij_i^*}-\mu_{1j_1^*} = (\mu_{ij_i^*}-\mu_{i1})+(\mu_{i1}-\mu_{11})+(\mu_{11}-\mu_{1j_1^*})> (-\delta_I) + \delta + 0 = \delta - \delta_I, 
\end{equation*}
thanks to \eqref{eqn:InnerIZ}. Hence, it suffices to take $\delta_O=\delta - \delta_I$.

In order to determine the values of $\delta_I$ and $\delta_O$ for a given $\delta$, we minimize the total sample size of the procedure over the choice of  $(\delta_I,\delta_O)$. Specifically, we solve the following optimization problem for each pair of systems
\begin{equation}\label{eqn:twostage_IZAllocation}
\begin{aligned}
& \underset{\delta_I,\delta_O>0}{\mbox{minimize}}&& \max\left\{\frac{h^2S_{ij,i'j'}^2}{\delta_I^2},\frac{h^2S_{ij,i'j'}^2}{\delta_O^2}\right\}\\[0.5ex]
& \mbox{subject to}  & & \delta_I +\delta_O= \delta
\end{aligned}
\end{equation}
It is straightforward to solve this problem and the optimal solution is $\delta_I=\delta_O=\delta/2$.

\begin{remark}
Notice that $\left\lceil h^2S_{ij,i'j'}^2/\delta_I^2\right\rceil$ (resp., $\left\lceil h^2S_{ij,i'j'}^2 /\delta_O^2\right\rceil$) represent the sample size required by the comparison between system $(i,j)$  and system $(i',j')$ in the inner-layer (resp., outer-layer) selection process. Since the inner-layer and outer-layer selection processes are conducted simultaneously after all the samples are generated, we should treat both layers equally and thus assign equal computational budget to them, leading to the optimal choice $\delta_I=\delta_O=\delta/2$.
\end{remark}

\subsection{Error Allocation}\label{sec:twostage_errorallocation}

Besides the IZ parameter, another critical parameter that determines the efficiency of the two stage RSB procedure is the error allowance $\beta$ for each necessary pairwise comparison. It must be chosen judiciously in order that the statistical validity \eqref{eqn:robust} be achieved.

First, we notice that 
$\mathbb{P}\{\mbox{CS} \} \geq \mathbb{P}\left\{\mbox{system $(1,1)$ is selected}\right\}$, since selecting system $(1,1)$ is obviously a CS event.
Then, the Bonferroni inequality can bound the right-hand-side from below, allowing us to give a PCS guarantee via choosing an appropriate $\beta$. 
Nevertheless, it is worthwhile to point out that due to the double-layer structure, CS can be obtained even if system $(1, 1)$ is eliminated, as long as the selected worst system of alternative 1 is better than any other alternative's selected worst system. This contributes the over-coverage of the realized PCS; see  Section \ref{sec:robust_numerical}. 

Specifically, system $(1,1)$ being eliminated by system $(1,j)$ amounts to $\bar{X}_{11}(N) < \bar{X}_{1j}(N)$, $j=2,\ldots,m$, whereas system $(1,1)$ being eliminated by system $(i,j)$ amounts to $\bar{X}_{11}(N) > \bar{X}_{ij}(N)=\max\limits_{1\leq l\leq m}\bar{X}_{il}(N)$, $i=2,\ldots,k$, $j=1,\ldots,m$. Then, it follows that 
\begin{eqnarray}
\mathbb{P}\{\mbox{ICS} \}
&\leq &\mathbb{P}\left\{\mbox{system $(1,1)$ is not selected}\right\} \nonumber\\
&=& \mathbb{P}\left\{\bigcup_{j=2}^m \left\{\bar{X}_{11}(N) <\bar{X}_{1j}(N)\right\}\;\bigcup\; \bigcup_{i=2}^k\bigcup_{j=1}^m  \left\{\bar{X}_{11}(N) > \bar{X}_{ij}(N)=\max_{1\leq l\leq m}\bar{X}_{il}(N)\right\} \right\}\nonumber\\
&\leq &\sum_{j=2}^m \mathbb{P}\left\{\bar{X}_{11}(N) < \bar{X}_{1j}(N) \right\} + \sum_{i= 2}^k\sum_{j=1}^m \mathbb{P}\left\{\bar{X}_{11}(N) > \bar{X}_{ij}(N) \right\},\label{eq:multi_rule}
\end{eqnarray}
where ICS is short for incorrect selection and the last inequality follows the Bonferroni inequality. Therefore, we can achieve a target PCS  $1-\alpha$ by ensuring each of the $km-1$ terms in the summation \eqref{eq:multi_rule} bounded by $\beta=\alpha/(km-1)$ from above. We name this method of error allocation the \textit{multiplicative rule}.

Nevertheless, the multiplicative rule can easily become over-conservative even if $k$ and $m$ are both moderate. For instance, if $k=m=10$, then the error allowance under the multiplicative rule is equivalent to that for the SB problem with $100$ alternatives. Observe that the over-conservativeness of the multiplicative rule stems from the fact that it uses the event of not selecting system $(1,1)$ to represent the ICS event itself. By doing so, we implicitly treat all the $km-1$ pairwise comparisons as equally important. In fact, we do not need to ensure correct selection of the worst system for each alternative in the inner-layer selection process, except for alternative 1. The bulk of the $km-1$ pairwise comparisons associated with the multiplicative rule turn out to be unnecessary. To see this, notice that
\begin{equation}\label{eq:add_rule1}
\mathbb{P}\{\mbox{ICS} \} = \mathbb{P}\left\{\bigcup_{i=2}^k\left\{\max_{1\leq j\leq m}\bar{X}_{1j}(N)>\max_{1\leq j\leq m}\bar{X}_{ij}(N)\right\}\right\} \leq \mathbb{P}\left\{A \cup B \right\},
\end{equation}
where $A = \bigcup_{i=2}^k\left\{\max_j\bar{X}_{1j}(N)>\max_j\bar{X}_{ij}(N)\right\}$ and $B=\bigcup_{j=2}^m\left\{\bar{X}_{11}(N)<\bar{X}_{1j}(N)\right\}$, and that
\begin{eqnarray}
A\cup B = (A\cap B^c)\cup B 
&=& \bigcup_{i=2}^k\left\{\bar{X}_{11}(N) >\max_{1\leq j\leq m}\bar{X}_{ij}(N)\right\} \;\cup\; B,\label{eq:add_rule2}
\end{eqnarray}
since $\max_j\bar{X}_{1j}(N) = \bar{X}_{11}(N)$ on $B^c$. 
Hence, by \eqref{eq:add_rule1} and \eqref{eq:add_rule2},
\begin{eqnarray}
\mathbb{P}\{\mbox{ICS} \}
&\leq&  \mathbb{P}\left\{\bigcup_{i=2}^k\left\{\bar{X}_{11}(N)>\max_{1\leq j\leq m}\bar{X}_{ij}(N)\right\}\;\bigcup\;\bigcup_{j=2}^m\left\{\bar{X}_{11}(N)<\bar{X}_{1j}(N)\right\}\right\} \nonumber\\
&\leq&  \mathbb{P}\left\{\bigcup_{i=2}^k\left\{\bar{X}_{11}(N)>\bar{X}_{i1}(N)\right\}\;\bigcup\;\bigcup_{j=2}^m\left\{\bar{X}_{11}(N)<\bar{X}_{1j}(N)\right\}\right\}\nonumber\\
&\leq& \sum_{i=2}^k\mathbb{P}\left\{\bar{X}_{11}(N)>\bar{X}_{i1}(N)\right\}+\sum_{j=2}^m\mathbb{P}\left\{\bar{X}_{11}(N)<\bar{X}_{1j}(N)\right\}.\label{eq:add_rule3}
\end{eqnarray}
The inequality above implies that there are $k+m-2$ ``critical" pairwise comparisons in the RSB problem. To achieve a target PCS $1-\alpha$, we can simply make each of the $k+m-2$ terms in the summation \eqref{eq:add_rule3} be no greater than $\beta=\alpha/(k+m-2)$. We name this method of error allocation the \textit{additive rule}. Since its total sample size is increasing in $\beta$, the two-stage RSB procedure with the additive rule is significantly more efficient than the one with the multiplicative rule.

Nevertheless, the additive rule is not applicable for the sequential RSB procedure developed in Section \ref{sec:robust_sequential} and the multiplicative rule will be used there. This is because the additive rule assumes implicitly that the worst system of each alternative is always retained during the inner-layer selection process, while it may be eliminated in early iterations of the sequential procedure; see Remark  \ref{remark:seq} for more discussion. 

\subsection{Statistical Validity}

We show that the two-stage RSB procedure equipped with the additive rule of error allocation is statistically valid. The case of the multiplicative rule can be proved similarly.

\begin{theorem}\label{theo:twostage}
Suppose that $\{X_{ij}: i=1,2,\ldots,k, j=1,2,\ldots,m\}$ are jointly normally distributed. Set the error allowance $\beta=\alpha/(k+m-2)$. Then, the two-stage RSB procedure is statistically valid, i.e., $\mathbb{P}\{\mu_{i^*1}-\mu_{11}\leq \delta\}\geq 1-\alpha$.
\end{theorem}

The two-stage RSB procedure selects the best alternative based on the means of typically large samples, which can be viewed as approximately normally distributed. To simplify theoretical analysis, we assume that the observations of the systems are jointly normally distributed. Then, Theorem \ref{theo:twostage} states that the two-stage RSB procedure has finite-sample statistical validity. We relax the normality assumption to allow the simulation outputs to have a general distribution for the sequential RSB procedure in next section at the expense of the finite-sample statistical validity. The sequential RSB procedure is statistically valid only asymptotically as the target PCS goes to 1.

\section{Sequential RSB Procedure}\label{sec:robust_sequential}

Sequential procedures for the SB problem typically require smaller sample sizes than two-stage procedures, because the former allow inferior systems to be eliminated dynamically during iterations  \citep{KimNelson01}. If switching between simulations of different systems does not incur substantial computational overhead, then the overall  efficiency of sequential procedures is usually much higher \citep{hong2005the}. 

Before presenting our sequential RSB procedure, we remark that there is a plain-vanilla sequential procedure for the RSB problem. One can simply apply a sequential SB procedure separately in each layer. Specifically, for each alternative,  a sequential SB procedure is applied to select its worst system; it is then applied again to the collection of ``worst systems'' to select the best among them.  However, this procedure has a major drawback: outer-layer eliminations occur only \textit{after} the worst system of each alternative is identified in the inner layer. This incurs excessive samples in the inner layer selection process, since it is unnecessary to identify the worst system for alternatives that are unlikely to be the best. By contrast, our sequential RSB procedures facilitates simultaneous elimination of all the surviving systems of an alternative when the alternative appears to be inferior with high likelihood. Our sequential RSB procedure is iterative with the following structure.

\begin{enumerate}[label=(\roman*)]
\item
Take an initial number of samples  to estimate the mean of each system and the variance of each pairwise difference.
\item\label{item:inner}
Perform the inner-layer selection: for each surviving alternative, eliminate systems that are unlikely to produce the worst-case mean performance.
\item
Perform the outer-layer selection: eliminate inferior alternatives based on the estimated worst-case mean performance of each surviving alternative.
\item
If there is only one surviving alternative or all the surviving alternatives are close enough (determined by the IZ parameter) to each other, then stop; otherwise, take one additional sample from each surviving system, update the statistics, and return to step \ref{item:inner}.
\end{enumerate}

For the inner-layer selection in step (ii), we apply the IZ-free sequential SB procedure in \cite{izfree2016}, hereafter referred to as the FHN procedure. This procedure does not require an IZ parameter and thus we can set the outer-layer IZ parameter to be the same as the overall IZ parameter. The reason for choosing the FHN procedure instead of other sequential SB procedures that based on the IZ formulation is because it is hard to construct an analytically tractable optimization problem, similar to \eqref{eqn:twostage_IZAllocation}, for the sequential RSB procedure that relates the decomposition of the overall IZ parameter to the efficiency of the procedure. See Section \ref{sec:sequential_inner} for details.

For the outer-layer selection in step (iii), a pairwise comparison between two surviving alternatives is done by constructing a confidence interval that bounds the difference between their worst-case mean performances. The confidence level of this interval depends on the error allowance $\beta$. If the confidence interval does not contain zero, then the two competing alternative are differentiated and the inferior one is eliminated (i.e. all the surviving systems of that alternative are eliminated simultaneously). See Section \ref{sec:sequential_outer} for details. 

\subsection{Inner-Layer: Eliminating Systems}\label{sec:sequential_inner}

The objective of the inner-layer selection process is to perform sequential screening to eliminate systems that are unlikely to produce the worst-case mean performance of each alternative.

We apply the FHN procedure to the systems of each alternative. In this procedure, the (normalized) partial-sum difference process between two systems is approximated by a Brownian motion with drift. We can then differentiate the two systems by checking if the drift of the Brownian motion is nonzero.    This is done by monitoring if the Brownian motion exits a well-designed continuation region, whose boundaries are formed by $\pm g_c(t)$ for $t\geq 0$, where $g_c(t) = \sqrt{[c+\log(t+1)](t+1)}$ for some carefully chosen constant $c$ that depends on the target PCS $1-\alpha$.

More specifically, consider alternative $i$ and let $\bar{X}_{ij}(n)$ denote the sample mean based on the first $n$ independent replications of system $(i,j)$. Define $t_{ij,ij'}(n) = n \sigma_{ij,ij'}^{-2}$ and $ Z_{ij,ij'}(n) = t_{ij,ij'}(n)[\bar{X}_{ij}(n)-\bar{X}_{ij'}(n)]$, 
where $\sigma_{ij,ij'}^2=\Var[X_{ij}-X_{ij'}]$, for any $1\leq j\neq j'\leq m$. Then, $Z_{ij,ij'}(n)$ can be approximated in distribution by a Brownian motion possibly with a nonzero drift. For any pairwise comparison between system $(i,j)$ and system $(i,j')$ with $j\neq j'$, we keep taking samples from them (i.e., increasing $n$) until
$|Z_{ij,ij'}(t_{ij,ij'}(n))|\geq g_c(t_{ij,ij'}(n))$, at which point the one with a smaller estimated mean performance is eliminated by the other since we are seeking the system that has the largest mean performance.
Once eliminated, a system will not be considered in any subsequent comparisons.

\subsection{Outer-Layer: Eliminating Alternatives}\label{sec:sequential_outer}

The sequential RSB procedure allows \textit{simultaneous} elimination of all the surviving systems of an alternative. This is achieved in the outer-layer selection process by comparing the estimated worst-case mean performances of the surviving alternatives. Notice that pairwise comparisons in the outer-layer are done for alternatives, instead of systems, and the comparisons are based on random sets of surviving systems of the two alternatives. As a result, sequential procedures for the SB problem are not applicable here.

Consider alternatives $i$ and $i'$ that have survived after $n$ samples of the relevant systems. To design an elimination rule between them, we construct a dynamic confidence interval $(L_{ii'}(n), U_{ii'}(n))$ for $\mu_{i1} - \mu_{i'1}$, the difference between their worst-case mean performances, i.e.,
\[\mathbb{P}\{\mu_{i1}-\mu_{i'1}\in (L_{ii'}(n), U_{ii'}(n)), \mbox{ for all } n<\infty\} \geq 1-\epsilon,\]
for a given confidence level $1-\epsilon$. Hence, if $L_{ii'}(n)>0$ (resp., $U_{ii'}(n)<0$), then $\mu_{i1}>\mu_{i'1}$ (resp., $\mu_{i1}<\mu_{i'1}$) with statistical significance and we eliminate alternative $i$ (resp., $i'$); otherwise, we continue sampling. In Proposition $\ref{prop:robustsequential_CB}$, we present an asymptotically valid approach for constructing such a confidence interval.

\begin{proposition}\label{prop:robustsequential_CB}
For $i=1,2,\ldots,k$, let $\mathcal{S}_i(n)$ denote the set of surviving systems of alternative $i$ after taking $n$ samples of the relevant systems and the subsequent inner-layer elimination. For $\beta\in(0,1)$, let $g_c(t)=\sqrt{[c+\log(t+1)](t+1)}$ with $c=-2\log(2\beta)$. For any two alternatives $i$ and $i'$, define an interval $(L_{ii'}(n), U_{ii'}(n))$ as follows
\begin{equation}\label{eqn:robustsequential_CB_outer}
\begin{aligned}
L_{ii'}(n) &= \max_{(i,j)\in \mathcal{S}_i(n)}\bar{X}_{ij}(n)-\max_{(i',j)\in \mathcal{S}_{i'}(n)}\bar{X}_{i'j}(n)-C_i(n)- D_{ii'}(n),\\
U_{ii'}(n) &= \max_{(i,j)\in \mathcal{S}_i(n)}\bar{X}_{ij}(n)-\max_{(i',j)\in \mathcal{S}_{i'}(n)}\bar{X}_{i'j}(n)+C_{i'}(n)+ D_{ii'}(n),
\end{aligned}
\end{equation}
where
\[C_i(n) = \max_{(i,j),(i,j')\in \mathcal{S}_i(n)} \frac{g_c(t_{ij,ij'}(n))}{t_{ij,ij'}(n)}\quad\mbox{and}\quad D_{ii'}(n) = \max_{(i,j)\in \mathcal{S}_i(n),(i',j')\in\mathcal{S}_{i'}(n)} \frac{g_c(t_{ij,i'j'}(n))}{t_{ij,i'j'}(n)}.\]
If $(i,1)\in \mathcal{S}_i(n)$ and $(i',1)\in \mathcal{S}_{i'}(n)$ for all $n\geq 1$, then
\begin{eqnarray*}
\limsup_{\beta\to 0}\ \frac{1}{2\beta} \mathbb{P}\left\{\mu_{i1}-\mu_{i'1}\notin \left(L_{ii'}(n),U_{ii'}(n)\right)\mbox{ for some } n\geq 1\right\}\leq 1.
\end{eqnarray*}
\end{proposition}

Under the IZ formulation, the sequential RSB procedure stops if either of the following conditions holds: (i) all but one alternatives are eliminated; (ii) all the surviving alternatives are sufficiently close to each other. The latter condition amounts to $C_{i}(n)+D_{ii'}(n)\leq \delta$ for any pair of surviving alternatives $i$ and $i'$ in the light of \eqref{eqn:robustsequential_CB_outer}. The above stopping criterion ensures that the unique best alternative or a good alternative is ultimately selected with certain statistical guarantee.

\subsection{The Procedure}\label{sec:robustsequential_proc}
We now present the sequential RSB procedure (Procedure S).

\vspace{5pt}

\begin{procedure}[Procedure S: Sequential RSB Procedure]\label{proc: robustsequential}
\begin{enumerate}
\item[]
\setcounter{enumi}{-1}
\item \textit{Setup.}
Specify the error allowance $\beta=\alpha/(km-1)$ and the first-stage sample size $n_0\geq 2$. Set $c=-2\log(2\beta)$.
\item \textit{Initialization.} Set $n=n_0$. Set $\mathcal S=\{1,2,\ldots,k\}$ to be  the set of surviving alternatives. Set $\mathcal{S}_i=\{(i,j): j=1,2,\ldots,m\}$ to be  the set of surviving systems of alternative $i$, $i=1,\ldots,k$. Take $n$ independent replications $X_{ij,1},\ldots, X_{ij,n}$ of each system $(i,j)$.
\item\label{item:update} \textit{Updating.} Compute the sample mean of each surviving system and the sample variance of the difference between each pair of surviving systems as follows
\[
\begin{array}{c}
\displaystyle\bar{X}_{ij}(n)=\frac{1}{n}\sum_{r=1}^{n} X_{ij,r}, \quad i\in\mathcal{S}, \; (i,j)\in \mathcal{S}_i,\\
\displaystyle S_{ij,i'j'}^2(n) = \frac{1}{n-1}\sum_{r=1}^{n}\left[X_{ij,r}-X_{i'j',r}-(\bar{X}_{ij}(n)-\bar{X}_{i'j'}(n))\right]^2,  \quad i,i'\in\mathcal{S}, \; (i,j)\in \mathcal{S}_i,\; (i',j')\in\mathcal{S}_{i'}.
\end{array}\]
\item \textit{Elimination.} For each $(i,j)\in \mathcal{S}_i$,  $(i',j')\in \mathcal{S}_{i'}$ with $i,i'\in \mathcal{S}$ and $i\neq i'$ or $j\neq j'$, compute
\[
\tau_{ij,i'j'}(n)=\frac{n}{S^2_{ij,i'j'}(n)} \quad\mbox{ and }\quad  Z_{ij,i'j'}(n) = \tau_{ij,i'j'}(n)[\bar{X}_{ij}(n)-\bar{X}_{i'j'}(n)].
\]
% Set $I^{\mbox{old}}=I$ and $I^{\mbox{old}}_i = \mathcal{S}_i$ with $i\in I$.
\begin{enumerate}[label*=\arabic*]
\item \textit{Inner-layer.} For each $i\in \mathcal{S}$, assign
\begin{align*}
\mathcal{S}_i \gets \mathcal{S}_i\setminus\{(i,j)\in \mathcal{S}_i: Z_{ij,ij'}(n)\leq -&g_c(\tau_{ij,ij'}(n)) \mbox{ for some } (i,j')\in \mathcal{S}_i\}.
\end{align*}

\item\textit{Outer-layer.} For each $i\in \mathcal{S}$, compute
\[C_i(n) = \max_{(i,j),(i,j')\in \mathcal{S}_i} \frac{g_c(\tau_{ij,ij'}(n))}{\tau_{ij,ij'}(n)};\]
for any other $i'\in \mathcal{S}$, compute
\[\tau_{ii'}^*(n)=\min_{(i,j)\in \mathcal{S}_i,(i',j')\in\mathcal{S}_{i'}}\tau_{ij,i'j'}(n)\quad\mbox{and}\quad W_{ii'}(n)=  \max_{(i,j)\in \mathcal{S}_i}\bar{X}_{ij}(n)-\max_{(i',j)\in \mathcal{S}_{i'}}\bar{X}_{i'j}(n).\]
Assign
\[\mathcal{S}\gets \mathcal{S}\setminus \{i\in \mathcal{S}: \tau_{ii'}^*(n)[W_{ii'}(n)-C_i(n)]> g_c(\tau_{ii'}^*(n))\mbox{ for some } i'\in \mathcal{S}\}.\]

\end{enumerate}

\item \textit{Stopping.} If either $|\mathcal{S}|=1$ or
\[\tau_{ii'}^*(n)[\delta-C_i(n)]\geq g_c(\tau_{ii'}^*(n))\mbox{ and }\tau_{ii'}^*(n)[\delta-C_{i'}(n)]\geq g_c(\tau_{ii'}^*(n)),\quad \mbox{for all $i,i'\in\mathcal{S}$},\]
then stop and select $i^*=\argmin\limits_{i\in\mathcal{S}} \max\limits_{(i,j)\in \mathcal{S}_i}\bar{X}_{ij}(n)$
as the best alternative. Otherwise, take one additional replication of each $(i,j)\in\mathcal{S}_i$ with $i  \in \mathcal{S}$, assign $n\gets n+1$, and return to step \ref{item:update}. 
\end{enumerate}
\end{procedure}

\subsection{Asymptotic Statistical Validity}\label{sec:seq_valid}

The FHN procedure for the SB problem allows the samples of the competing alternatives to have a general distribution at the expense of the finite-sample statistical validity. We show that the sequential RSB procedure equipped with the multiplicative rule of error allocation is statistically valid in an asymptotic regime in which the targeted PCS level goes to 1 (i.e., $1-\alpha\to 1$). This regime is adopted by \cite{izfree2016} and dates back to \cite{perng1969comparison} and \cite{Dudewics1969an}.

\begin{theorem}\label{theo:sequential}
Suppose that $\{X_{ij}: i=1,2,\ldots,k, j=1,2,\ldots,m\}$ are generally distributed, and that  the moment generating function of $\{X_{ij}: i=1,2,\ldots,k, j=1,2,\ldots,m\}$ is finite in a neighborhood of the origin of $\mathbb R^{k\times m}$. 
Let $n_0(\alpha)$ denote the initial sample size of the sequential RSB procedure as a function of $\alpha$.  Set the error allowance $\beta=\alpha/(km-1)$. If $n_0(\alpha)\to \infty$ as $\alpha\to 0$, then the sequential RSB procedure is statistically valid asymptotically as $\alpha\to 0$, i.e.,  $\limsup\limits_{\alpha\to 0}\mathbb{P}\{\mu_{i^*1}-\mu_{11}> \delta\}/\alpha\leq 1$.
\end{theorem}

\begin{remark}
Theorem \ref{theo:sequential} indicates that with generally distributed simulation outputs, the sequential RSB procedure is statistically valid for all $\alpha>0$ small enough. The assumption on $n_0$ is imposed to make sure that the distribution of the sample means converges to the normal distribution as $\alpha\to0$. This assumption ensures the strong consistency of sequentially updated variance estimators, thereby facilitating asymptotic analysis of the procedure. Nevertheless, numerical experiments show that even with a moderate $n_0$, the procedure can still deliver a pre-specified PCS.
\end{remark}

\begin{remark}\label{remark:seq}
The additive rule of error allocation does not apply here. In fact, inequality \eqref{eq:add_rule3}, which underpins the additive rule, implicitly assumes that the worst system of each alternative is in contention in the inner-layer selection process. This is certainly the case for the two-stage RSB procedure. By contrast, the sequential RSB procedure has multiple rounds of inner/outer-layer selection, and the worst system of an alternative may be eliminated in an early round. Notice that the dynamic confidence interval \eqref{eqn:robustsequential_CB_outer} hinges on the condition that the worst system of each alternative is never eliminated in the inner-layer. Indeed, if the system having the largest mean performance of an alternative is eliminated in the inner-layer, the remaining systems of the same alternative will yield a smaller estimate of the worst-case mean performance, which makes this alternative less likely be eliminated in the outer-layer. Therefore, in order that the sequential RSB procedure be statistically valid, the errors associated with all the pairwise comparisons among the systems of each alternative must be controlled, which makes the multiplicative rule necessary.

\end{remark}

\section{Computational Efficiency}\label{sec:robust_numerical}
In this section we focus on the procedures' efficiency for a small $\alpha$ through a set of numerical experiments that generalize standard tests for SB procedures.

Suppose that there are $k\times m$ systems, where system $(i,j)$ refers to the pair of alternative $i$ and $j$th probability scenario in the ambiguity set. Let $X_{ij}$ denote the random performance of system $(i,j)$, for $i=1,2,\dots,k, j=1,2,\dots,m$ and suppose that $(X_{ij}: i=1,2,\ldots,k,\ j=1,2,\ldots,m)$ are mutually independent normal random variables, $X_{ij}\sim \mathcal{N}(\mu_{ij},\sigma_{ij}^2)$. Under Assumption \ref{asp:basic}, the two RSB procedures aim to select alternative 1 upon termination in an attempt to solve $\min_i\max_j \mathbb{E}[X_{ij}]$. In particular, we consider two different configurations of the means that generalize the \textit{slippage configuration} (SC) and the \textit{monotone decreasing means} (MDM) configuration for SB procedures. For SC we use
\begin{equation}\label{eqn:sc}
[\mu_{ij}]_{k\times m} =
\begin{pmatrix}
0 && 0 && \ldots && 0\\
0.5 && 0.5 && \ldots && 0.5\\
\vdots && \vdots && \ddots && \vdots \\
0.5 && 0.5 && \ldots && 0.5
\end{pmatrix},
\end{equation}
and for MDM we use
\begin{equation}\label{eqn:mdm}
[\mu_{ij}]_{k\times m} =
\Big(0.5(i-1)-0.2(j-1)\Big)_{1\leq i\leq k,1\leq j\leq m}
\end{equation}
Notice that with the configuration \eqref{eqn:sc}, the outer-layer selection process deals with the largest means of each row, i.e., $(0,0.5,0.5,\ldots,0.5)$, which is a SC for the SB problem with IZ formulation. With the configuration \eqref{eqn:mdm}, the means of each row are monotonically decreasing so the inner-layer selection process for each row corresponds to a SB problem with IZ formulation and a MDM configuration; the outer-layer selection process, on the other hand, deals with $(0, 0.5, \ldots, 0.5(k-1))$, also a monotone configuration of means.

Independently of the means, we further consider three configurations of the variances:
\begin{enumerate}[label=(\roman*)]
\item
Equal-variance (EV) configuration: $\sigma_{ij}^2=1$ for all $i,j$.
\item
Increasing-variance (IV) configuration: $\sigma_{ij}^2 = (1+0.1(i-1))(1+0.1(j-1))$ for all $i,j$.
\item
Decreasing-variance (DV) configuration: $\sigma_{ij}^2 = (1+0.1(i-1))^{-1} (1+0.1(j-1))^{-1}$ for all $i,j$.
\end{enumerate}

In all the experiments below we set the initial sample size $n_0=10$ and the target PCS $1-\alpha=0.95$. For each experiment specification (i.e., IZ parameter, values of $k$ and $m$, configuration of the means and the variances, rule of error allocation, RSB procedure), we repeat the experiment 1000 times independently. We find that the realized PCS is 1.00 for all the cases. It is well known in SB literature that selection procedures that rely on the Bonferroni inequality usually over-deliver PCS \citep{Frazier14}. Thus, we only report the average sample size of each procedure in this section.

\subsection{Comparison Between Multiplicative Rule and Additive Rule}\label{sec:robustnumerical_errorallocation}

We set the IZ parameter $\delta=0.5$ and decompose it equally into the inner-layer and outer-layer IZ parameters. Notice that the average sample size of Procedure T does not depend on the means. We study the impact of error allocation on Procedure T's efficiency by varying the problem scale (i.e., $k$ and $m$) and configuration of the variances. The numerical results are presented in Figure \ref{fig:two-stage}.

\begin{figure}[t]
\begin{center}
\caption{Average Sample Sizes of Procedure T with Different Error Allocation Rules}\label{fig:two-stage}
\includegraphics[width=0.8\textwidth]{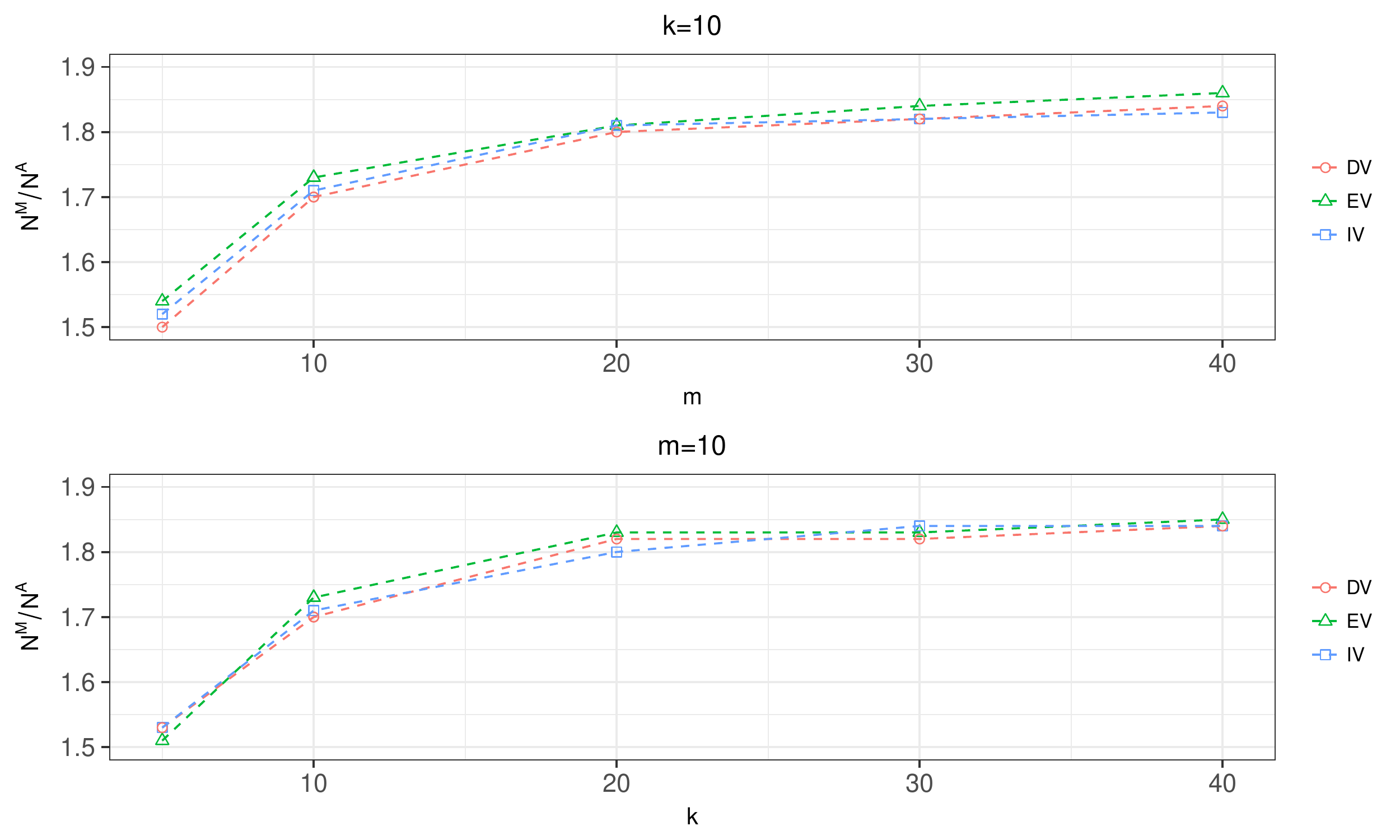}
\end{center}
\small{\textit{Note.} Top: $m$ varies with $k=10$; bottom: $k$ varies with $m=10$. The vertical axis represents the ratio of the average sample size of Procedure T under the multiplicative rule to that under the additive rule.}
\end{figure}

First, as expected the multiplicative rule requires more samples than the additive rule, regardless of problem scale or configuration of the means/variances. In all the experiments, $N^{\mathrm{M}}$ is about $50\% \sim 90\%$ greater than $N^{\mathrm{A}}$, where $N^{\mathrm{M}}$ and $N^{\mathrm{A}}$ denote the average sample size under the multiplicative and additive rules, respectively.

Second, for a given problem scale, the ratio $N^{\mathrm{M}}/N^{\mathrm{A}}$ is almost independent of configuration of the variances. This is because with an equal split of $\delta$ into the inner-layer and the outer-layer, the average sample size is $N\approx 4h^2\delta^{-2}\max_{i,j}\sigma_{ij}^2$, where $h=t_{1-\beta,n_0-1}$ is the $100(1-\beta)\%$ quantile of the student distribution with $n_0-1$ degrees of freedom. Hence, for any given configuration of the variances, $N^{\mathrm{M}}/N^{\mathrm{A}}\approx h_{\mathrm{M}}^2/h_{\mathrm{A}}^2$. So essentially the ratio is  determined by the value of $\beta$, which depends on the rule of error allocation and problem scale for a given target PCS.

Third, for the same reason, given a configuration of the variances the ratio $N^{\mathrm{M}}/N^{\mathrm{A}}$ increases as $k$ or $m$ increases. However, the rate at which this ratio increases is low relative to the increase in the values of $k$ and $m$. This is caused by the fact that the standard normal quantile grows very slowly towards the tail of the distribution.

\subsection{Comparison Between Procedure T and Procedure S}

We have argued in Remark \ref{remark:seq} that the additive rule does not apply to Procedure S. We now compare Procedure S under the multiplicative rule with Procedure T under the additive rule, and show that the former is much more efficient.

The two RSB procedures are implemented under different combinations of problem scale, IZ parameter, and configuration of the means/variances. Again, we choose $\delta$ from 0.5, 0.25 and 0.1, but only present the numerical results for EV configuration (seen in Figure \ref{fig:sequential_twostage}) because the results for the other two configurations of the variances are very similar.

\begin{figure}[t]
\begin{center}
\caption{Average Sample Sizes of Procedure T and Procedure S Under the EV Configuration }\label{fig:sequential_twostage}
\includegraphics[width=0.8\textwidth]{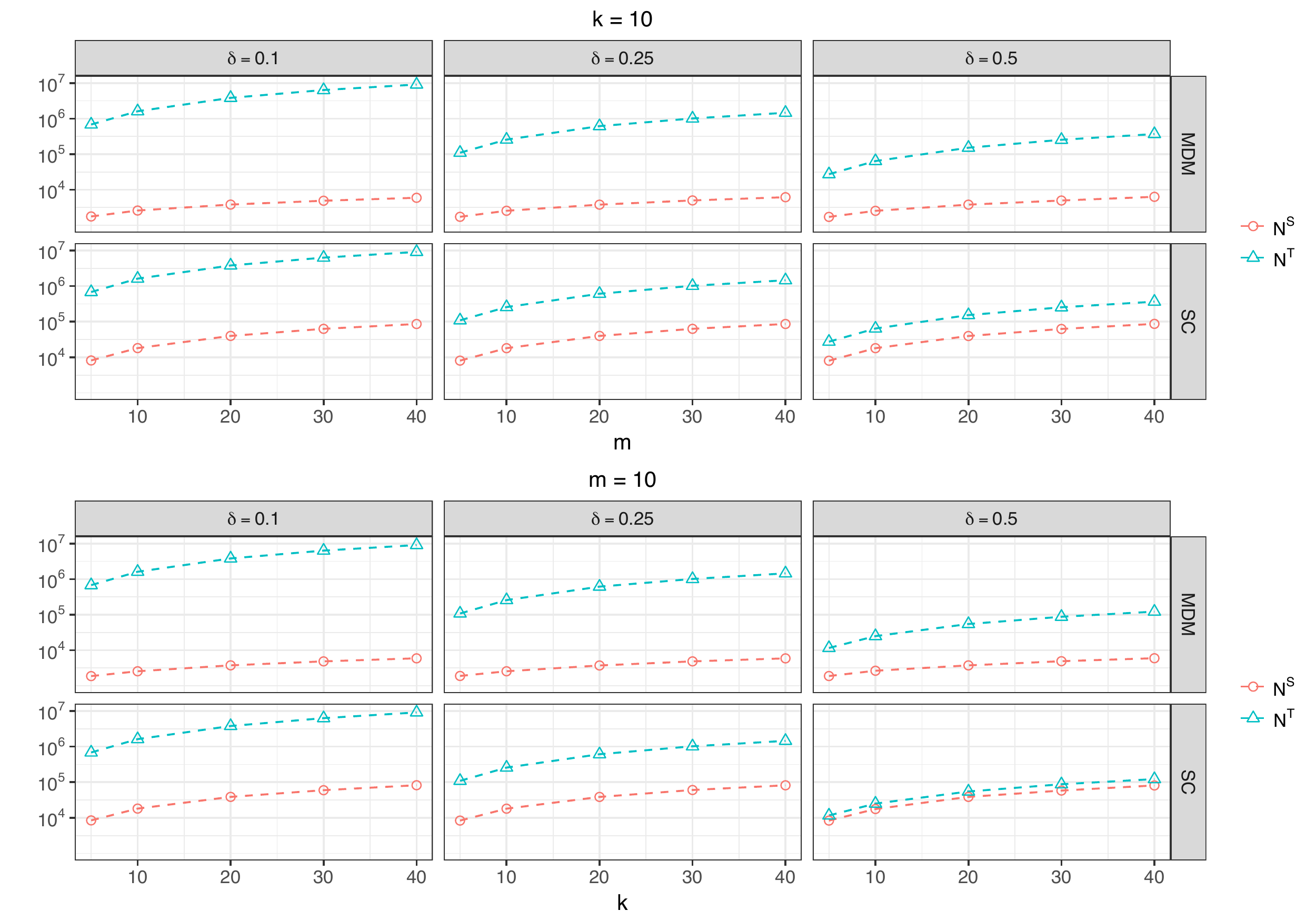}
\end{center}
\small{\textit{Note.} Top: $m$ varies with $k=10$; Bottom: $k$ varies with $m=10$. The vertical axis is on a logarithmic scale.}
\end{figure}

First, Procedure S requires a dramatically  fewer samples than Procedure T, seen from the fact that $N^{\mathrm{S}}$ is much smaller than $N^{\mathrm{T}}$, where $N^{\mathrm{S}}$ and $N^{\mathrm{T}}$ denote the average sample sizes of Procedure S and Procedure T, respectively. The difference can be orders of magnitude for large-scale problems with a small IZ parameter; see, e.g., the subplots in the first column of Figure \ref{fig:sequential_twostage}. This suggests that there are an enormous number of early eliminations of systems/alternatives in Procedure S.

Second, interestingly, $N^S$ is not sensitive to $\delta$ when $\delta$ is small relative to the difference between the best and the second-best alternatives, i.e., $\delta\leq\mu_{21}-\mu_{11}$. For example, the dashed lines with circles in the three subplots in the first row of Figure \ref{fig:sequential_twostage} are almost identical to each other. 
Recall that Procedure S terminates if either of the following two conditions is met: (i) there is only one alternative left, and (ii) all the surviving alternatives differ from each other in terms of their worst-case mean performance by no more than $\delta$. The low sensitivity of $N^{\mathrm{S}}$ with respect to $\delta$ for small $\delta$ suggests that if $\delta\leq\mu_{21}-\mu_{11}$, the procedure terminates primarily because the first stopping condition is met, and thus the procedure mostly selects the unique best alternative rather than a good alternative between several that it cannot differentiate. 
This feature makes Procedure S particularly attractive in practice. Knowing little of the differences between competing alternatives, a decision-maker tends to choose a small $\delta$ to make sure that the unique best alternative is identified. Hence, when using Procedure S, the decision-maker can specify an arbitrarily small $\delta$ without any increase in computational cost. 
By contrast, $N^T$ is roughly proportional to $\delta^{-2}$.

Third, unlike Procedure T, the efficiency of Procedure S does depend on the configuration of the means; see, e.g., the first two subplots in the first column of Figure \ref{fig:sequential_twostage}. With everything else the same, SC requires substantially more samples than MDM does. In SC all the systems of each alternative have the same mean performance, and all alternatives but the best one have the same worst-case mean performance. By contrast, in MDM all the systems of each alternative differ from one another by a clear margin and so do the alternatives. Hence, both inner-layer and outer-layer eliminations in Procedure S with SC occur significantly less than with MDM.

Last, it turns out that as the problem scale (i.e., $km$) increases or as $\delta$ decreases, $N^{\mathrm{T}}/N^{\mathrm{S}}$ in general increases. Hence, the advantage of Procedure S in terms of efficiency is more significant for problems with a larger scale or for a smaller IZ parameter; see the numerical results in Section \ref{EC_sec:TS}. Moreover, such an increase in the ratio is more substantial for MDM than for SC. This is a consequence of of the last finding -- it is more difficult for Procedure S to conduct eliminations with SC than with MDM.

\subsection{Comparison Between Procedure S and Plain-Vanilla Sequential RSB Procedure}\label{sec:vanilla}
We have argued at the beginning of Section \ref{sec:robust_sequential} that a major advantage of Procedure S relative to a plain-vanilla sequential RSB procedure is that it does not need to identify the worst system of each alternative, and can perform simultaneous elimination of all the surviving systems of an alternative if the alternative is unlikely to be the best, which can save a substantial number of samples. 

Notice that virtually any sequential SB procedure can be used to form a plain-vanilla sequential RSB procedure. To ensure a fair comparison against Procedure S whose inner-layer selection is performed by the FHN procedure, we use Procedure 3 in \cite{izfree2016}, which is a truncated version of the FHN procedure that has an additional stopping criterion to guarantee termination of the procedure in finite time in the case of two or more alternatives having the same mean performance. 
We shall call this plain-vanilla sequential RSB procedure Procedure V; see the detailed specification in Section \ref{EC_sec:SV} of Appendix.

Both Procedure S and Procedure V are implemented under different combinations of problem scale, IZ parameter, and configuration of the means/variances. Notice that the difference between the worst-case mean performance of the best and the second-best alternatives is 0.5 in both configurations of the means \eqref{eqn:sc} and \eqref{eqn:mdm}. We choose $\delta$ from 0.5, 0.25, and 0.1. The numerical experiment shows that Procedure S requires significantly fewer samples than Procedure V in general. One exception, however, is that if the configuration of the means is SC and $\delta=\mu_{21}-\mu_{11}=0.5$, then the computational costs of Procedure S and Procedure V are almost identical. In this case, the alternatives are hard to differentiate in early iterations of Procedure S, so simultaneous elimination of the surviving systems of an alternative that is unlikely to be the best rarely happens. The numerical results and more discussion are presented in Section \ref{EC_sec:SV} of Appendix.

\section{A Multiserver Queue with Abandonment}\label{sec:queueing}

The theory in this paper is developed under the premise that the ambiguity set is known and fixed. In practice, however, an ambiguity set is typically constructed based on available data, and thus it may vary significantly. Therefore, it is important to study the potential impact of data variation on usefulness of the RSB approach for simulation-based decision-making. Nevertheless, to characterize theoretically such impact is beyond the scope of this paper. Instead, we present an extensive numerical investigation in the context of queueing simulation.

Consider a $G/G/s+G$ model, i.e., a queueing system that has $s$ identical servers and allows a customer to abandon the system before receiving any service if her waiting time in the queue is deemed to exceed her patience. The interarrival times, service times, and patience times are independent and generally distributed. Suppose that both the interarrival time and the patience time have a known distribution, but the distribution of the service time $P_0$ is unknown. Instead, a finite sample from $P_0$ is available for constructing an ambiguity set $\mathcal{P}$.

For customer $i$, $i=1,\ldots,n$, let $I_i$ and $W_i$ denote her indicator for abandonment and the waiting time, respectively. Let $\xi$ denote the service time. We measure the quality of service by both the probability of abandonment and the mean waiting time of those customers who do not abandon the system. In particular, consider the following cost function 
\begin{equation}\label{eq:queue_cost}
f(s, \xi) = c_A U\left(n^{-1} N_A \right) + c_W (n-N_A)^{-1}\sum_{\{i: I_i =0 \}} W_i + c_Ss,
\end{equation}
where $N_A=\sum_{i=1}^n I_i$ is the (random) number of customers who abandon the system, $U(\cdot)$ is a utility function, and $c_A$, $c_W$, and $c_S$ are all positive constants. Let $s$ be the decision variable and assume that it takes values from  $\{1,2,\ldots,k\}$. To minimize the worst-case mean cost over the ambiguity set, we solve 
$\min\limits_{1\leq s\leq k} \max\limits_{P\in\mathcal{P}} \mathbb{E}[f(s, \xi)]$,
which becomes a RSB problem of the form  \eqref{eq:minimax}.

Recent empirical studies on service times in various service industries, including telephone call centers \citep{BGMSSZZ05} and health care \citep{StrumMayVargas00}, show that the lognormal distribution often fits historical data well. Hence, we assume that $P_0$ is the lognormal distribution with mean 1; equivalently, $\log(\xi)$ is normally distributed with mean $-\sigma^2/2$ and variance $\sigma^2$. The ambiguity set $\mathcal{P}$ is constructed as follows. Upon observing a sample from $P_0$, we use the maximum likelihood estimation (MLE) to fit three distribution families (lognormal, gamma, and Weibull) to the sample. Then, we conduct the Kolmogorov-Smirnov (K-S) test to each fitted distribution, and include in $\mathcal{P}$ those that are not rejected by the test at significance level 0.05.

Other related parameters are specified as follows. We assume that the interarrival time and the patience time are both exponentially distributed with mean 0.1 and 5, respectively. We set the largest number of servers $k=10$, the length of each sample path $n=10,000$, the utility function $U(p) = \log(1/(1-p))$, the constants $c_A=4$, $c_W=2$, $c_S=1$, the first-stage sample size $n_0=10$, and the target PCS $1-\alpha = 0.95$. We vary $\sigma\in\{1,2,3\}$,  corresponding to different extents of variability of the service time. (The coefficient of variation of $\xi$ is 1.31, 7.32, 90.01, respectively, in the three cases.) We set the sample size $\ell$  to be either 50 or 500 to imitate the scenarios of having a small and large amount of data to construct the ambiguity set, respectively.

We now assess usefulness of the RSB approach for a decision maker who does not know the input distribution but has input data to work with. The benchmark approach is a typical one in practice: fit a group of distribution families to the data and use the ``best fitted'' distribution as if it were the true distribution, discarding the others some of which may be plausible as well. We call this approach the best-fitting (BF) approach. We simply define the best fitted distribution as the one having the smallest K-S statistic among all the fitted distributions. Notice that the RSB approach is reduced to the BF approach if the ambiguity set contains only the best fitted distribution. 

We first generate 10,000 samples of the cost function $f(s, \xi)$ for each $s=1,\ldots,k$, with $\xi$ following the true distribution $P_0$. For each $s$, we estimate based on the samples the mean and the $p$-quantile of $f(s, \xi)$ under $P_0$, denoted by $M$ and $Q^p$, respectively. They are used as performance measures for evaluating $s_{\mathrm{RSB}}$ and $s_{\mathrm{BF}}$, the decisions obtained by  the RSB and the BF approach, respectively. We also compare the RSB approach with clairvoyance, i.e., the decision maker knows the true distribution and makes the decision $s_{\mathrm{Tr}}$  that minimizes $\mathbb{E}[f(s,\xi)]$ under $P_0$. Let $M_{\mathrm{Tr}}$ and $Q^p_{\mathrm{Tr}}$ denote the corresponding performance measures. Albeit impractical, clairvoyance provides a lower-bound on the mean cost that other approaches could possibly achieve. 

Then, we conduct 1,000 macro-replications of the following experiment. 
\begin{enumerate}[label=(\roman*)]
\item 
Generate a sample of service times from $P_0$.
\item 
Construct an ambiguity set $\mathcal{P}$ based on the sample. 
\item 
Run a RSB procedure on $\mathcal P$ and $\{\hat P_0\}$ to obtain solutions $s_{\mathrm{RSB}}$ and $s_{\mathrm{BF}}$, respectively.  
\item 
Retrieve the performance measures of $s_{\mathrm{RSB}}$ (resp., $s_{\mathrm{BF}}$)  computed earlier under $P_0$.  
\item 
Compute the relative difference between $s_{\mathrm{Tr}}$ and $s_{\mathrm{RSB}}$ and that between   $s_{\mathrm{BF}}$ and $s_{\mathrm{RSB}}$ in terms of the performance measures. 
\end{enumerate}
Hence, there are 1,000 realizations of each relative difference above. Their average values over the macro-replications are reported in Table \ref{tab:rel_diff_tr} and Table \ref{tab:rel_diff}. We also estimate the probability that the MLE-fitted lognormal distribution is rejected by the K-S test. It turns out to be 0 based on the 1,000 macro-replications for each case ($\sigma=1,2,3$, $\ell =50, 500$). Hence, the ambiguity sets constructed in our experiments always include the true distribution family; see also Table \ref{tab:rel_diff_mis} for their average size.

\begin{table}[ht]
\begin{center}
\caption{Relative Differences (in \%) Between Clairvoyance and RSB}\label{tab:rel_diff_tr}
\begin{tabular}{cccccccccccc}
\toprule
\multirow{2}{*}{$\sigma$}              && \multirow{2}{*}{Sample Size}  && \multicolumn{7}{c}{Relative Difference}\\
\cmidrule{5-11}
             &&       && $\frac{M_{\mathrm{Tr}}}{M_{\mathrm{RSB}}}-1$     && $\frac{Q^{0.7}_{\mathrm{Tr}}}{Q^{0.7}_{\mathrm{RSB}}}-1$     && $\frac{Q^{0.8}_{\mathrm{Tr}}}{Q^{0.8}_{\mathrm{RSB}}}-1$    && $\frac{Q^{0.9}_{\mathrm{Tr}}}{Q^{0.9}_{\mathrm{RSB}}}-1$      \\
\midrule
\multirow{2}{*}{$1$} && $50$  && $-1.72_{\pm 0.13}$ && $-1.71_{\pm 0.13}$ && $-1.70_{\pm 0.14}$ && $-1.69_{\pm 0.14}$ \\
                   && $500$ && $-0.40_{\pm 0.05}$ && $-0.39_{\pm 0.05}$ && $-0.38_{\pm 0.05}$ && $-0.38_{\pm 0.05}$ \\
\midrule
\multirow{2}{*}{$2$} && $50$  && $-3.55_{\pm 0.32}$ && $-3.77_{\pm 0.34}$ && $-3.88_{\pm 0.35}$ && $-4.07_{\pm 0.38}$ \\
                   && $500$ && $-0.71_{\pm 0.07}$ && $-0.88_{\pm 0.05}$ && $-1.00_{\pm 0.08}$ && $-1.17_{\pm 0.10}$ \\
\midrule                   
\multirow{2}{*}{$3$} && $50$  && $-7.35_{\pm 0.48}$ && $-7.03_{\pm 0.51}$ && $-6.75_{\pm 0.54}$ && $-6.30_{\pm 0.59}$ \\
                   && $500$ && $-1.36_{\pm 0.11}$ && $-1.21_{\pm 0.12}$ && $-1.09_{\pm 0.14}$ && $-0.93_{\pm 0.17}$ \\
\bottomrule
\end{tabular}
\end{center}
\small{\textit{Note.} ``$\pm$'' indicates 95\% confidence interval.}
\end{table}

\begin{table}[ht]
\begin{center}
\caption{Relative Differences (in \%) Between BF and RSB}\label{tab:rel_diff}
\begin{tabular}{cccccccccccc}
\toprule
\multirow{2}{*}{$\sigma$}              && \multirow{2}{*}{Sample Size}  && \multicolumn{7}{c}{Relative Difference}\\
\cmidrule{5-11}
             &&       && $\frac{M_{\mathrm{BF}}}{M_{\mathrm{RSB}}}-1$     && $\frac{Q^{0.7}_{\mathrm{BF}}}{Q^{0.7}_{\mathrm{RSB}}}-1$     && $\frac{Q^{0.8}_{\mathrm{BF}}}{Q^{0.8}_{\mathrm{RSB}}}-1$    && $\frac{Q^{0.9}_{\mathrm{BF}}}{Q^{0.9}_{\mathrm{RSB}}}-1$      \\
\midrule
\multirow{2}{*}{$1$} && $50$  && $0.12_{\pm 0.07}$ && $0.14_{\pm 0.08}$ && $0.15_{\pm 0.08}$ && $0.17_{\pm 0.08}$  \\
                   && $500$ && $0.00_{\pm 0.04}$ && $0.00_{\pm 0.04}$ && $0.00_{\pm 0.04}$ && $0.00_{\pm 0.04}$ \\
\midrule
\multirow{2}{*}{$2$} && $50$  && $1.64_{\pm 0.32}$ && $1.76_{\pm 0.33}$ && $1.84_{\pm 0.34}$ && $1.96_{\pm 0.35}$  \\
                   && $500$ && $0.02_{\pm 0.04}$ && $0.01_{\pm 0.04}$ && $0.01_{\pm 0.04}$ && $0.01_{\pm 0.05}$  \\
\midrule                   
\multirow{2}{*}{$3$} && $50$  && $5.32_{\pm 0.92}$ && $5.66_{\pm 0.93}$ && $5.90_{\pm 0.95}$ && $6.23_{\pm 0.98}$  \\
                   && $500$ && $0.07_{\pm 0.12}$ && $0.08_{\pm 0.12}$ && $0.09_{\pm 0.12}$ && $0.11_{\pm 0.12}$ \\
\bottomrule
\end{tabular}
\end{center}
\end{table}

We compare the RSB approach with the clairvoyance based on Table \ref{tab:rel_diff_tr}. First, as expected, the RSB approach yields a higher mean cost than clairvoyance. The right tail quantiles are also higher for the RSB approach. The gap between the performance measures of the two approaches reflects the decision maker's lack of information about the true distribution. Second, with everything else the same, the gap increases as $\sigma$ increases. This is because a larger $\sigma$ means a larger stochastic variability in the service times, in which case the decision maker is more uncertain about the true distribution, thereby paying a larger penalty for the deeper uncertainty. Third, in the same vein, the gap decreases as the sample size $\ell$ increases, since the input uncertainty can be greatly reduced by a large amount of data. Indeed, the relative differences in absolute value are about 1\% or even lower when $\ell =500$.

Assuming clairvoyance is obviously not practical. It is more interesting for practitioners to compare the RSB approach with the BF approach, which is almost common practice for simulation-based decision-making. From Table \ref{tab:rel_diff}, we first find that the RSB solution outperforms, or performs at least as good as, the BF solution for all the performance measures. Relying on worst-case analysis, the RSB approach is conservative by design and should protect the decision-maker against extreme cases, producing reliable performance even if the true distribution is not in her favor. Therefore, the RSB solution performing better for the right tail quantiles is expected. It is somewhat surprising, however, that the RSB solution performs better for the mean cost as well. This suggests that in the presence of input uncertainty, the potential risk that the BF approach ends up with a misspecified input distribution can be so significant that it overwhelms the price that the decision maker needs to pay for being conservative.  Second, the advantage of the RSB approach over the BF approach is larger when the uncertainty about the true distribution is deeper, which means either the true distribution has a larger stochastic variability (larger $\sigma$), or the sample size $\ell$ is smaller. In particular, with $\sigma=3$ and  $\ell = 50$, the RSB approach outperforms the BF approach by a significant margin ($5\% \sim 6\%$). Third, notice that with $\ell = 500$, the two approaches  deliver nearly identical performance. This is because with a large $\ell$, the input uncertainty is marginal and the ambiguity set mostly consists of only the best fitted distribution, since the possibility that the fitted gamma or Weibull distribution is not rejected by the K-S test is nearly zero in this situation.

\begin{table}[ht]
\begin{center}
\caption{Relative Differences Between BF and RSB Conditional on Model Misspecification} \label{tab:rel_diff_mis}
\begin{tabular}{ccccccccc}
\toprule
\multirow{2}{*}{$\sigma$} & \multirow{2}{*}{$\ell$} && \multirow{2}{*}{$\mathbb{E}(|\mathcal P|)$}   
&   \multirow{2}{*}{$\mathbb{P}(\mathrm{misspec.})$ (in \%) } &   \multicolumn{4}{c}{Rel. Diff. (in \%) } \\
\cmidrule{6-9}
& && & & $\frac{M_{\mathrm{BF}}}{M_{\mathrm{RSB}}}-1$     & $\frac{Q^{0.7}_{\mathrm{BF}}}{Q^{0.7}_{\mathrm{RSB}}}-1$     & $\frac{Q^{0.8}_{\mathrm{BF}}}{Q^{0.8}_{\mathrm{RSB}}}-1$    & $\frac{Q^{0.9}_{\mathrm{BF}}}{Q^{0.9}_{\mathrm{RSB}}}-1$    \\
\midrule
$1$     & $50$  &&  $2.93_{\pm 0.02}$ & $17.90_{\pm 2.38}$ &  $\phantom{1}0.39_{\pm 0.26}$  & $\phantom{1}0.45_{\pm 0.26}$  & $\phantom{1}0.48_{\pm 0.26}$  & $\phantom{1}0.53_{\pm 0.26}$  \\
$2$     & $50$  && $2.56_{\pm 0.03}$ & $17.40_{\pm 2.35}$ &  $\phantom{1}9.05_{\pm 1.38}$  & $\phantom{1}9.70_{\pm 1.35}$  & $10.14_{\pm 1.34}$ & $10.84_{\pm 1.33}$ \\
$3$     & $50$  && $2.27_{\pm 0.03}$ & $17.30_{\pm 2.34}$ & $31.28_{\pm 3.09}$ & $33.22_{\pm 2.88}$ & $34.56_{\pm 2.78}$ & $36.41_{\pm 2.64}$ \\
\bottomrule
\end{tabular}
\end{center}
\small{\textit{$\mathbb{E}(|\mathcal P|)$: the mean size of the ambiguity set.}}
\end{table}

We now further discuss the cause of the difference in performance between the RSB approach and the BF approach. When $\ell$ is not large, the best fitted distribution is likely to be gamma or Weibull instead of lognormal, which we refer to as \textit{model misspecification}. We estimate the probability of model misspecification and compute the relative differences between the two approaches conditional on model misspecification. The results are shown in Table \ref{tab:rel_diff_mis}. (The estimated probability of model misspecification is 0  for $\sigma=1,2,3$ if $\ell= 500$, so we do not include them in the table. )

Table \ref{tab:rel_diff_mis} shows that first, the probability of model misspecification is fairly large ($17\%\sim18\%$ on average), if $\ell$ is not large.  Second, if the best fitted distribution is misspecified, the consequence for the BF approach can be severe, resulting in a cost that can be over 30\% higher than the cost of the RSB solution. Notice that the ambiguity set contains $2\sim 3$ plausible probability distributions on average. This suggests that in the case of model misspecification, both the true and the incorrectly chosen distributions inform the decision produced by the RSB approach. Third, in the case of model misspecification, the relative difference in performance between the two approaches grows dramatically (from less than 1\% to over 30\%) as $\sigma$ increases. Since the estimated probability of model misspecification is roughly the same for different values of $\sigma$, the characteristics of the queueing system under the wrongly chosen input model is clearly the cause of the significant degradation in performance of the BF approach as $\sigma$ grows. Therefore, the RSB approach has a substantial advantage over the BF approach for protecting the decision-maker against model misspecification, especially when the input uncertainty is deep.

\section{An Appointment-Scheduling Problem}\label{sec:scheduling}

Appointment-scheduling problems are ubiquitous in the healthcare industry; see \cite{GuptaDenton08} for a comprehensive survey. One challenge in addressing these problems in practice is that the appointment duration is often random and its distribution is hard to estimate due to lack of the data. In the light of the distributional uncertainty, robust optimization has recently emerged as a popular framework for solving this class of problems \citep{KongLeeTeoZheng13,Mak2015appointment,Qi17}. In this section, we apply our RSB approach to study an appointment-scheduling problem with real data and compare it with existing approaches.

We consider the appointment-scheduling problem in \cite{Mak2015appointment}. There are $n$ operations to be performed by $n$ different surgeons in an operating  room during a time interval $[0, T]$. By operation $i$, we mean the operation performed by surgeon $i$, $i=1,\ldots,n$. Operation $i$ requires a random duration $d_i$ to complete, which follows distribution $P_i$. Let $\mathcal P_i$ denote the ambiguity set for $P_i$. We assume that $P_i$'s are mutually independent, so the ambiguity set for the joint distribution $\mathbf{P}=(P_1,\ldots,P_n)$ can be expressed as the Cartesian product of $(\mathcal P_1,\ldots,\mathcal P_n)$, i.e., $\mathcal{P} = \mathcal P_1\times \cdots\times\mathcal P_n$.  

Let $\psi$ be a permutation of $\{1,2,\ldots,n\}$ that indicates the sequence of the operations performed in the operating room. Let $t_i$ denote the time allowance of operation $i$. The planner makes two decisions: the sequence of the operations $\psi$ and the time allowance of each operation $\mathbf{t}=(t_1,\ldots,t_n)$. 

Suppose that all the appointments are \textit{scheduled} to be completed by time $T$, so the feasible region of  $\mathbf{t}$ is $\mathcal{T}=\{\mathbf {t}\in \mathbb{R}_+^n:\sum_{i=1}^nt_i\leq T\}$. If an operation does not start as planned because of delay of completion of the previous operation, a waiting cost is incurred at the rate of $c_W$; if the last operation is completed after time $T$, an overtime cost is incurred at the rate of $c_O$. Let $W_i$ denote the waiting time of operation $\psi_i$, $i=1,\ldots,n$, and $W_{n+1}$ denote the overtime.
% ; see Figure \ref{fig:schedule} for an illustration. 
Then, $W_1=0$ and 
$W_i = \max\{0, W_{i-1}+d_{\psi_{i-1}}-t_{\psi_{i-1}}\}$, $i=2,\ldots,n+1$.
Hence, letting  $\mathbf{d}=(d_1,\ldots,d_n)$, the total waiting and overtime cost as a function of  $(\psi,  \mathbf{d}, \mathbf{t})$ is 
$f(\psi,  \mathbf{d}, \mathbf{t}) = c_W\sum_{i=1}^nW_i +c_OW_{n+1}$.
To minimize the worst-case mean cost over the ambiguity set, the planner solves 
\begin{equation}\label{eq:scheduling}
 \min_\psi \min_{\mathbf{t}\in\mathcal{T}} \max_{\mathbf{P}\in\mathcal{P}} \mathbb{E} [f(\psi,  \mathbf{d}, \mathbf{t})].
\end{equation}

In \cite{Mak2015appointment} each ambiguity set $\mathcal{P}_i$ is the set of distributions having the same first two moments as those of $P_i$, which are assumed to be known. Then, problem \eqref{eq:scheduling} for a given sequence $\psi$ can be reformulated as a second-order conic program that has an analytical solution. Specifically, the optimal time allowance is of the form $\tilde t_{\psi_i}^*=\mu_{\psi_i}+\eta_{\psi_i} \sigma_{\psi_i}$, where $\mu_i$ and $\sigma_i$ are the mean and standard deviation of $P_i$, and $\eta_i$ is a constant that depends on $\{i, \psi, (\mu_1,\sigma_1),\ldots, (\mu_n,\sigma_n)\}$ and can be computed analytically. Theorem 3 in their paper further shows that the optimal sequence $\tilde \psi^*$ follows the increasing order of variances (OV), i.e., $\sigma_{\tilde\psi^*_1}\leq \cdots \leq \sigma_{\tilde \psi^*_n}$.

However, their approach does not apply in our setting where we formulate $\mathcal P_i$ as a finite set of plausible probability distributions. In this study, we use the same rule of determining the time allowances as \cite{Mak2015appointment}  for a given sequence, and focus on the optimal sequencing problem. The scheduling problem \eqref{eq:scheduling} is then reduced to 
\begin{equation}\label{eq:surgery_RSB}
\min_\psi \max_{\mathbf{P}\in\mathcal{P}} \mathbb{E} [f(\psi,  \mathbf{d}, \mathbf{\tilde t}^*_{\psi})],
\end{equation}
which becomes a RSB problem of the form \eqref{eq:minimax}, where $\mathbf{\tilde t}^*_{\psi} = (\tilde t^*_{\psi_1},\ldots,\tilde t^*_{\psi_n})$. 

We now solve \eqref{eq:surgery_RSB} with the ambiguity set $\mathcal P$ constructed based on real data from a hospital in Anhui Province, China. This hospital is a major healthcare facility in the province. \footnote{This hospital has around 2,800 beds. According to Becker's Hospital Review, the largest hospital in the United States has around 2,400 beds as of 2015.} We use the data on durations of cesarean sections in the hospital that were performed in 2014. We fix $n=4$ as a representative scenario in the hospital, so the number of competing alternatives  is $4!=24$. 

We consider two datasets: one is relatively large, denoted by $\mathcal D^L$, and consists of 4 obstetricians (OBs) having 138, 97, 84, and 68 cases in the record, respectively; the other is relatively small, denoted by $\mathcal D^S$, and consists of another 4 OBs having 66, 60, 55, and 54 cases, respectively. 

Given a dataset $\mathcal D$ (either $\mathcal D^L$ or $\mathcal D^S$),  let $\mathcal D_i$ denote the observations of $d_i$ for OB $i$, $i=1,\ldots,4$. Since the distribution that generates the real data is unknown, we assume that the ``true'' distribution of $d_i$ is the empirical distribution based on $\mathcal D_i$. Let $\mu_i$ and $\sigma_i^2$ denote the mean and variance of this distribution. We stress here that the RSB approach as well as other approaches introduced later do not have access to $(\mu_i,\sigma_i^2)$. Instead, they only have access to a random sample $\mathcal F_i$  from $\mathcal D_i$. 

For comparison, we also apply the following approaches to find a proper sequence of OBs. 
\begin{itemize}
\item
\textit{Best-Fitting} (BF):  For $i=1,\ldots,n$, let $\hat P_i$ be the best fitted distribution for $\mathcal F_i$. Solve  \eqref{eq:surgery_RSB} with $\mathcal P_i=\{\hat P_i\}$ to obtain a sequence $\psi_{\mathrm{BF}}$. 
\item
\textit{Empirical} (Em): For each $i=1,\ldots,n$, let  $\hat P_i$ be the empirical distribution based on $\mathcal F_i$. Solve  \eqref{eq:surgery_RSB} with $\mathcal P_i=\{\hat P_i\}$ to obtain a sequence $\psi_{\mathrm{Em}}$.
\item
\textit{OV}: Sort the OBs in the increasing order of $\hat\sigma_i^2$ to obtain a sequence $\psi_{\mathrm{OV}}$.
\end{itemize}

We conduct 1,000 macro-replications of the following experiment. 
\begin{enumerate}[label=(\roman*)]
\item 
Randomly select a fraction $\gamma\in(0,1)$ of $\mathcal D_i$, denoted by ${\mathcal F}_i$.
\item 
Construct an ambiguity set $\mathcal P_i$ by using MLE to fit six widely used distributions (exponential, gamma, Weibull, lognormal, Pareto, and triangular) to ${\mathcal F}_i$, and retaining the fitted distributions that are not rejected by the K-S test at significance level 0.05. 
\item 
Compute the mean and variance of $\mathcal F_i$, denoted by $(\hat\mu_i,\hat\sigma_i^2)$. Apply Theorem 2 of \cite{Mak2015appointment} to compute the time allowances $\tilde {\mathbf{t}}^*_\psi$ with $(\hat\mu_i,\hat\sigma_i^2)$, $i=1,\ldots,4$ for each sequence $\psi$. 
\item 
Run the four competing approaches to solve \eqref{eq:surgery_RSB}. %Denote the solutions to be $\psi_{\mathrm{RSB}}$, $\psi_{\mathrm{BF}}$, $\psi_{\mathrm{Em}}$, and $\psi_{\mathrm{OV}}$, respectively.
\item 
For each $\psi =\psi_{\mathrm{RSB}}, \psi_{\mathrm{BF}}, \psi_{\mathrm{Em}}, \psi_{\mathrm{OV}}$, generate $10^7$ samples of the cost function $f(\psi,  \mathbf{d}, \mathbf{\tilde t}^*_{\psi})$ under the ``true'' distribution of $d_i$. Compute $M$ and $Q^p$, $p=0.7,0.8,0.9$, based on the samples. 
\end{enumerate}

The other  parameters involved in our experiment are specified as follows: $c_W=1.0$, $c_O = 0.5$, $T=\sum_{i=1}^4\mu_i$, the first-stage sample size $n_0=10$, the IZ parameter $\delta = 1.0$, and the target PCS $1-\alpha = 0.95$. Moreover, we vary the fraction $\gamma\in\{0.2, 0.5, 0.8\}$ to imitate the scenarios of having a small, medium, and large amount of available data to construct the ambiguity set, respectively.

\begin{remark}
We find that $|\mathcal P_i|$ is typically 3 or 4, so the size of the ambiguity set $m=|\mathcal P|=\prod_{i=1}^4|\mathcal P_i|$  is fairly large, which typically ranges from 100 to 200 depending on $\mathcal D$ and $\gamma$. 
\end{remark}

We obtain 1,000 realizations of $M$ and $Q^{p}$, one from each macro-replication. Figure \ref{fig:hospital_mean} illustrates the distribution of $M$ for different values of $(\psi, \gamma, \mathcal D)$.
We omit the figures for $Q^{0.7}$, $Q^{0.8}$, and $Q^{0.9}$ since they are qualitatively similar to Figure \ref{fig:hospital_mean}. In addition, we compute the relative difference in $M$ or  $Q^p$ between a competing approach and the RSB approach: 
$\frac{ M_A }{M_{\mathrm{RSB}}} - 1$ and $\frac{Q^p_A }{Q^p_{\mathrm{RSB}} } - 1$,
where the subscript $A$ denotes the competing approach, i.e., BF, Em, or OV, for $p=0.7,0.8,0.9$. Their average over the macro-replications is shown in Figure \ref{fig:rel_diff}.

\begin{figure}[t]
\begin{center}
\caption{Distribution of Realizations of $M$}  \label{fig:hospital_mean}
\includegraphics[width=0.8\textwidth]{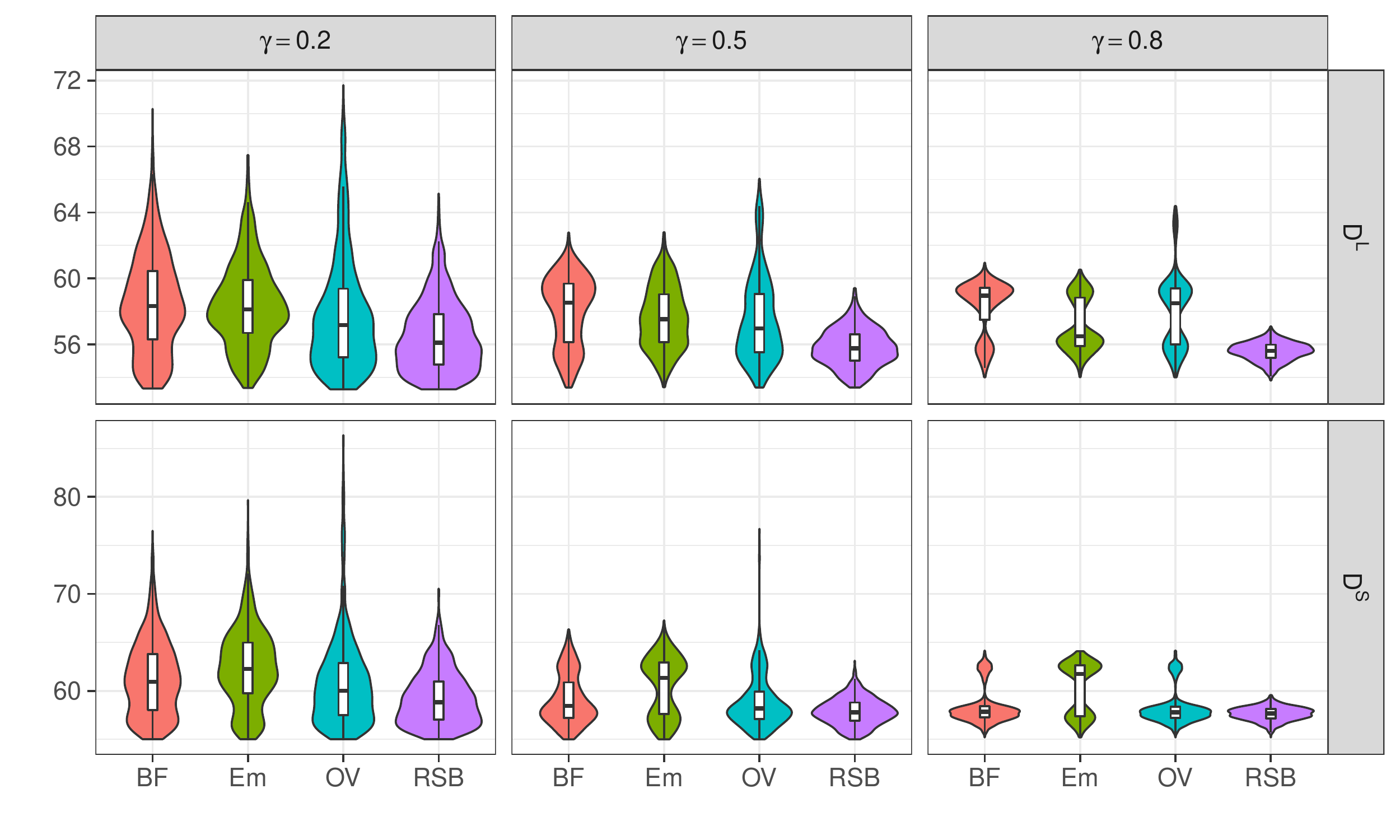}
\end{center}
\small{\textit{Note.} The violin plots depict the distributions of the presented data constructed via kernel density estimation; the box plots indicate the median, 25\% quantile, and 75\% quantile. }
\end{figure}

We have several findings. 
First,  the RSB approach consistently outperforms the other three approaches by a clear margin, producing the lowest mean cost and the lowest right tail quantiles.

Second, Figure \ref{fig:hospital_mean} also shows as $\gamma$ or equivalently the sample size increases, the range of the realizations of the mean cost for each approach obviously decreases. This is because the larger sample size is, the more information it provides about the ``true'' distribution. In addition, as input data varies, the mean cost of the RSB solution is the most stable, indicated by the range of its realizations being considerably narrower than the others. Notice that the sample size for constructing an ambiguity set in this example is not large, ranging from 11 to 110 depending on $\mathcal D$ and $\gamma$, and therefore the input uncertainty is substantial. In particular, the sample $\mathcal F_i$ may be far from being representative of the ``true'' distribution, i.e., the whole dataset $\mathcal D_i$. In this case, both the BF approach and the Empirical approach suffer from using a misspecified distribution for decision-making, while  the RSB approach can effectively protect the planner against such a risk. 

Third, the OV approach is well known to be a good heuristic rule in appointment-scheduling literature. It is essentially built upon worst-case analysis as well. Indeed, \cite{Mak2015appointment} proves that the OV sequence is optimal for the robust optimization problem \eqref{eq:scheduling}, provided that the ambiguity set $\mathcal P=\{(P_1,\ldots,P_n): \mathbb{E}(P_i) = \hat \mu_i, \mathrm{Var}(P_i) = \hat\sigma_i^2, i=1,\ldots,n\}$. Therefore, the OV approach also provides protection against input uncertainty, outperforming both the BF approach and the Empirical approach in all cases expect $(\mathcal D,\gamma)=(\mathcal D^L, 0.8)$; see Figure \ref{fig:rel_diff}. Nevertheless, its performance is significantly worse than the RSB approach. (The relative difference in $Q^{0.7}$ and $Q^{0.8}$ can be over 10\%.) This is because (i) the OV approach uses only the first two moments of the data, whereas the RSB approach uses more information by fitting various distribution families to the data; and (ii) the moments are estimated from data and the estimation error can be substantial if the sample size is not large, whereas fitting distributions to data appears to be more resilient to such error.  

In \cite{KongLeeTeoZheng16} several artificial examples of the scheduling problem are constructed for which the OV sequence is not optimal. They assume that the input distribution is known and find that the performance gap between the OV sequence and the optimal sequence is marginal. Our experiments, however, demonstrate that such performance gap may be sizable due to input uncertainty. In the light of the fact that lack of data is a common challenge for appointment-scheduling in health care \citep{macario2010}, our findings are practically meaningful, as they may encourage hospital administrators to adopt conservative formulations such as the RSB approach that do not assume mean and variance are correctly estimated from tens of data points, and to follow more established statistical procedures to accept or reject representative distributions.

\begin{figure}[t]
\begin{center}
\caption{Relative Differences Between Competing Methods and RSB} \label{fig:rel_diff}
\includegraphics[width=0.8\textwidth]{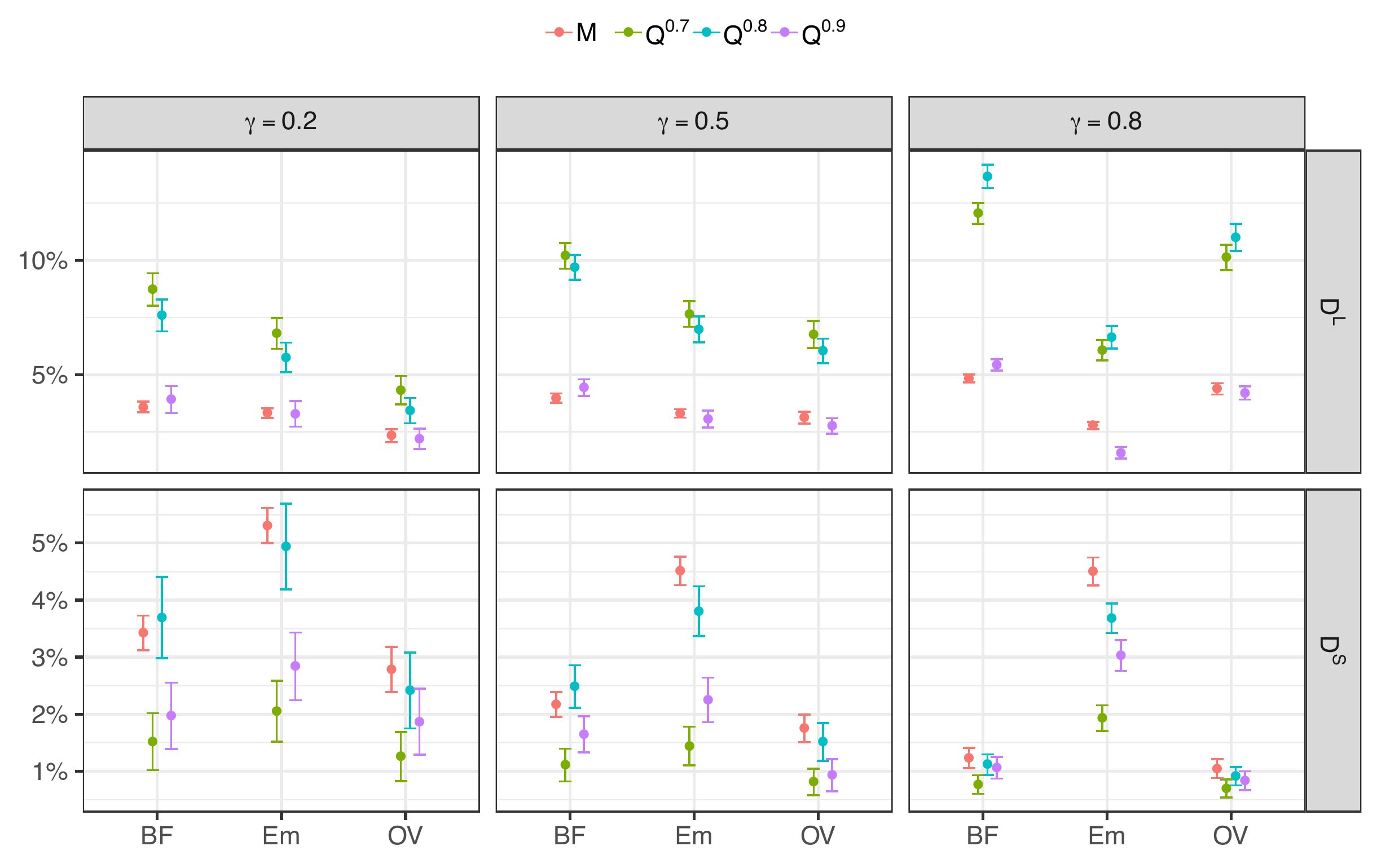}
\end{center}
\small{\textit{Note.} $M$ means $\frac{M_A}{M_{\mathrm{RSB}}}-1$ and $Q^p$ means  $\frac{Q^p_A}{Q^p_{\mathrm{RSB}}}-1$, for $A=\mathrm{BF}, \mathrm{Em}, \mathrm{OV}$, and $p=0.7,0.8,0.9$. The dots denote the average over the 1,000 macro-replications, whereas the error bars denote the 95\% confidence intervals.}
\end{figure}

\section{Concluding Remarks}\label{sec:conclusion}

We propose to use worst-case analysis to address the SB problem with input uncertainty. Two selection procedures are developed to solve the resulting RSB problem and are shown to be statistically valid in the finite-sample and asymptotic regimes, respectively. The two-stage RSB procedure is simple but rather conservative, requiring an excessive number of samples. The sequential RSB procedure, on the other hand, is not as easy to implement, but requires a dramatically smaller sample size even than the plain-vanilla sequential RSB procedure. 

As shown by both the queueing example in Section \ref{sec:queueing} and the scheduling example in Section \ref{sec:scheduling}, not only can the RSB  approach generate decisions that are reliable in extreme cases, but also perform better than the best-fitting approach on average. This makes the RSB approach a useful tool when the simulation model suffers from deep input uncertainty. In practice, given limited input data one may consider to apply both the RSB approach and the best-fitting (BF) approach to solve the SB problem in the presence deep input uncertainty, and use the latter as a numerical check to ensure that the former works as intended. Another possible way of checking the robustness of the RSB approach in practice is to numerically evaluate and compare the mean performance of the RSB decision under several plausible input distributions included in the ambiguity set.

This paper focuses on the uncertainty in specifying the parametric family of the input distribution, which motivates the critical assumption that the ambiguity set consists of finitely many distributions. This form of ambiguity set entails the double-layer structure of the two selection procedures. However, there are a variety of other ways to construct an ambiguity set, e.g., via moment constraints \citep{delage2010distributionally} or via statistical divergence \citep{ben2013robust}. Conceivably, changing the form of the ambiguity set would dramatically alter the structure of the RSB problem, and new selection procedures would need to be developed.

Parameter uncertainty is an equally important issue. Ideally, it should be addressed in conjunction with the uncertainty about the parametric family. One possible approach is to use likelihood ratio to characterize parameter uncertainty. Then, the worst-case mean performance over the uncertain parameters within the same family can be estimated by simulation under a ``nominal'' distribution only; see \cite{HuHong15} for details. Then, the ambiguity set that characterizes the uncertainty about both the parametric family and its parameters can be reduced to a finite set. We leave the exploration of these questions to future research.

\appendix
\section{Proof for the Two-Stage RSB Procedure}
\begin{proof}[Proof of Theorem \ref{theo:twostage}.]
We first notice that if $\mu_{k1}-\mu_{11}\leq\delta$, then by Assumption \ref{asp:basic} $\mu_{i1}-\mu_{11}\leq \delta$ for all $i=1,\ldots,k$, which implies that $\mathbb{P}(\mu_{i^*1}-\mu_{11}\leq \delta) = 1$ and the theorem trivially holds. Hence, without loss of generality we assume that there exists $l=1,\ldots,k-1$ for which $\mu_{l1}-\mu_{11}\leq \delta$ and $\mu_{l+1,1}-\mu_{11}>\delta$. Then, a good selection (i.e., $\{\mu_{i^*1}-\mu_{11}\leq \delta\}$) occurs if any alternative $i$, $i=l+1,\ldots,k$ is not selected. It follows that
\begin{eqnarray}\label{eqn:robustadd_proof1}
&& \mathbb{P}\{\mu_{i^*1}-\mu_{11}\leq\delta\} \nonumber \\
&\geq &\mathbb{P}\left\{\bigcap_{i=l+1}^k\left\{\max_{1\leq j\leq m} \bar{X}_{ij}(N)>\max_{1\leq j\leq m}\bar{X}_{1j}(N)\right\}\right\}\nonumber\\
&\geq &\mathbb{P}\left\{\bigcap_{i=l+1}^k \left\{\max_{1\leq j\leq m} \bar{X}_{ij}(N)>\max_{1\leq j\leq m}\bar{X}_{1j}(N)\right\}\bigcap \left\{\max_{1\leq j\leq m} \bar{X}_{1j}(N) - \bar{X}_{11}(N)<\delta_I\right\}\right\}\nonumber\\
&\geq & \mathbb{P}\left\{\bigcap_{i=l+1}^k \left\{\bar{X}_{i1}(N)>\max_{1\leq j\leq m}\bar{X}_{1j}(N)\right\}\bigcap \left\{\max_{1\leq j\leq m} \bar{X}_{1j}(N) - \bar{X}_{11}(N)<\delta_I\right\}\right\}\nonumber\\
&\geq &\mathbb{P}\left\{\bigcap_{i=l+1}^k\left\{\bar{X}_{i1}(N)>\bar{X}_{11}(N)+\delta_I\right\}\;\bigcap\;\bigcap_{j=2}^m\left\{\bar{X}_{1j}(N) - \bar{X}_{11}(N)<\delta_I\right\}\right\}\nonumber\\
&\geq &1-\sum_{i=l+1}^k\mathbb{P}\left\{\bar{X}_{i1}(N)\leq\bar{X}_{11}(N)+\delta_I\right\}-\sum_{j=2}^m\mathbb{P}\left\{\bar{X}_{1j}(N)\geq\bar{X}_{11}(N)+\delta_I\right\},
\end{eqnarray}
where the last step is due to the Bonferroni inequality. For each $i=l+1,\ldots,k$,
\begin{eqnarray}\label{eqn:robustmulti_proof4}
\mathbb{P}\{\bar{X}_{i1}(N)\leq \bar{X}_{11}(N)+\delta_I\} &=&\mathbb{P}\left\{\bar{X}_{i1}(N)-\bar{X}_{11}(N)-(\mu_{i1}-\mu_{11})\leq -\mu_{i1}+\mu_{11}+\delta_I\right\}\nonumber\\
&\leq &\mathbb{P}\left\{\bar{X}_{i1}(N)-\bar{X}_{11}(N)-(\mu_{i1}-\mu_{11})\leq -\delta_O\right\} \nonumber\\
&= &\mathbb{P}\left\{\frac{\bar{X}_{i1}(N)-\bar{X}_{11}(N)-(\mu_{i1}-\mu_{11})}{\sqrt{S_{11,i1}^2/N}} \leq -\frac{\delta_O}{\sqrt{S_{11,i1}^2/N}}\right\}\nonumber\\
&\leq &\mathbb{P}\left\{\frac{\bar{X}_{i1}(N)-\bar{X}_{11}(N)-(\mu_{i1}-\mu_{11})}{\sqrt{S_{11,i1}^2/N}}\leq -h\right\},
\end{eqnarray}
where the first inequality holds because $\delta_I+\delta_O=\delta$ and $\mu_{i1}-\mu_{11}>\delta$ for each $i=l+1,\ldots,k$ under Assumption \ref{asp:basic}, and the second inequality holds because $N\geq h^2S_{11,i1}^2/\delta_O^2$.

For notational simplicity, we suppress its dependence on $i$ and set $Y_r = X_{i1,r}-X_{11,r}-(\mu_{i1}-\mu_{11})$ for $r=1,\ldots,N$,  $\sigma_Y^2=\Var[Y_r]$, and $S_Y^2 = S_{11,i1}^2$. Applying \eqref{eqn:robustmulti_proof4} and following a similar derivation in \cite{stein1945}, 
\begin{eqnarray}\label{eqn:robustmulti_proof5}
\mathbb{P}\{\bar{X}_{i1}(N)\leq \bar{X}_{11}(N)+\delta_I\} \leq  \mathbb{P}\left\{\frac{\sum_{r=1}^{N} Y_r}{\sqrt{N S_Y^2}}\leq -h\right\} = \mathbb{P}\{Z\leq -h\} = \beta,
\end{eqnarray}
where $Z$ has the distribution of Student's $t$ with $n_0-1$ degrees of freedom, and the second inequality holds due to the definition of $h$. Likewise, we can show that
\begin{eqnarray}\label{eqn:robustmulti_proof3}
\mathbb{P}\left\{\bar{X}_{1j}(N)\geq\bar{X}_{11}(N)+\delta_I\right\}\leq \beta,
\end{eqnarray}
for each $j=2,\ldots,m$. Combining \eqref{eqn:robustadd_proof1}, \eqref{eqn:robustmulti_proof5}, and \eqref{eqn:robustmulti_proof3}, we have
\begin{eqnarray*}
\mathbb{P}\{\mu_{i^*1}-\mu_{11}>\delta\} 
& \leq &   \sum_{i=2}^k\mathbb{P}\left\{\bar{X}_{i1}(N)\leq\bar{X}_{11}(N)+\delta_I\right\}
+\sum_{j=2}^m\mathbb{P}\left\{\bar{X}_{1j}(N)\geq\bar{X}_{11}(N)+\delta_I\right\} \\
&\leq & (k+m-2)\beta = \alpha,
\end{eqnarray*}
which completes the proof.
\end{proof}

\section{Proofs for the Sequential RSB Procedure}
We prove Proposition \ref{prop:robustsequential_CB} and Theorem \ref{theo:sequential} in this section. To facilitate their proofs, we first present a result that characterizes the first-exit probability that a random walk exits from the region $(-g_c(t), g_c(t))$.

\begin{lemma}\label{lem:FirstExitProb}
Let $\{Y_n: i=n=1,2,\ldots\}$ be a sequence of independent and identically distributed (i.i.d.) random variables with mean 0 and variance $\sigma^2<\infty$. Let $\bar{Y}(n)$ and $S^2(n)$ denote the sample mean and the sample variance of $(Y_i:i=1,\ldots,n)$, respectively. Let $g_c(t)=\sqrt{[c+\log(t+1)](t+1)}$ with $c=-2\log(2\beta)$ for some $\beta\in(0,1)$. Assume that the moment generating function of $Y_1$ is finite in a neighborhood of zero. 
\begin{enumerate}[label=(\roman*)]
\item  Define  $t(n)=n/\sigma^2$ and $N = \min\{n\geq n_0: t(n)|\bar{Y}(n)|\geq g_c(t(n))\}$. If $n_0\to\infty$ as $\beta\to 0$, then,
\[\limsup\limits_{\beta\to 0}\frac{1}{\beta}\mathbb{P}\left\{t(N)\bar{Y}(N)\leq - g_c(t(N)), N<\infty\right\}\leq 1.\]

\item Define $\tau(n)=n/S^2(n)$ and $N' = \min\{n\geq n_0: \tau(n)|\bar{Y}(n)|\geq g_c(\tau(n))\}$. If $n_0\to\infty$ as $\beta\to 0$, then,
\[\limsup\limits_{\beta\to 0}\frac{1}{\beta}\mathbb{P}\left\{\tau(N')\bar{Y}(N')\leq -g_c(\tau(N')), N'<\infty\right\}\leq 1.\]
\end{enumerate}
\end{lemma}

\begin{proof}[Proof of Lemma \ref{lem:FirstExitProb}.] We provide a proof sketch here. The complete proof can be found in the proof for Theorem 2 in \cite{izfree2016}. Let $(B(t):t\geq 0)$ be a standard Brownian motion. By virtue of the functional central limit theorem \citep[p.102]{whitt2002stochastic}, it can be shown that
\[\mathbb{P}\left\{t(N)\bar{Y}(N)\leq - g_c(t(N)), N<\infty\right\} \leq \mathbb{P}\left\{B(T_c)\leq - g_c(T_c), T_c<\infty \right\}, \]
where $T_c=\inf\{t\geq 0: |B(t)|\geq g_c(t)\}$. Moreover, Example 6 in \cite{jennen1981first} shows that
\[\mathbb{P}\left\{B(T_c)\leq - g_c(T_c), T_c<\infty \right\} = \frac{1}{2}e^{-c/2} = \beta,\]
and thus (i) follows immediately. On the other hand, note that
\[
\limsup\limits_{\beta\to 0}\frac{1}{\beta}\mathbb{P}\left\{\tau(N')\bar{Y}(N')\leq -g_c(\tau(N')), N'<\infty\right\} \leq \limsup\limits_{\beta\to 0}\frac{1}{\beta}\mathbb{P}\left\{t(N)\bar{Y}(N)\leq -g_c(t(N)), N<\infty\right\},\]
which proves (ii).
\end{proof}

\subsection{Proof of Proposition \ref{prop:robustsequential_CB}}
\begin{proof}[Proof of Proposition \ref{prop:robustsequential_CB}.]
If $(i,1)\in \mathcal{S}_i(n)$, then
\begin{eqnarray}
U_{ii'}(n) &=& \max_{(i,j)\in \mathcal{S}_i(n)}\bar{X}_{ij}(n)-\max_{(i',j)\in \mathcal{S}_{i'}(n)}\bar{X}_{i'j}(n)+C_{i'}(n)+ D_{ii'}(n)\nonumber\\[0.5ex]
&\geq & \bar{X}_{i1}(n)-\max_{(i',j)\in \mathcal{S}_{i'}(n)}\bar{X}_{i'j}(n)+C_{i'}(n)+ \frac{g_c(t_{i1,i'1}(n))}{t_{i1,i'1}(n)}.\label{eqn:CBproof_1}
\end{eqnarray}
Let $(i', j^*_{i'})$ denote the system having the largest mean performance among $\mathcal{S}_{i'}(n)$. If $(i',1)\in \mathcal{S}_{i'}(n)$, then system $(i',1)$ has not been eliminated by $(i', j^*_{i'})$. According to the mechanism of the inner-layer elimination and its related discussion in Section \ref{sec:sequential_inner},
\begin{eqnarray}\label{eqn:CBproof_2}
\max_{(i',j)\in \mathcal{S}_{i'}(n)}\bar{X}_{i'j}(n)- \bar{X}_{i'1}(n) = \bar{X}_{i'j^*_{i'}}(n) - \bar{X}_{i'1}(n)< \frac{g_c(t_{i'1,i'j^*_{i'}}(n))}{t_{i'1,i'j^*_{i'}}(n)} \leq  C_{i'}(n).
\end{eqnarray}
Plugging \eqref{eqn:CBproof_2} into \eqref{eqn:CBproof_1} yields
\begin{eqnarray}\label{eqn:CBproof_3}
U_{ii'}(n)\geq \bar{X}_{i1}(n)-\bar{X}_{i'1}(n) +\frac{g_c(t_{i1,i'1}(n))}{t_{i1,i'1}(n)}.
\end{eqnarray}
Likewise, we can show that
\begin{eqnarray}\label{eqn:CBproof_4}
L_{ii'}(n)\leq \bar{X}_{i1}(n)-\bar{X}_{i'1}(n)-\frac{g_c(t_{i1,i'1}(n))}{t_{i1,i'1}(n)}.
\end{eqnarray}
By \eqref{eqn:CBproof_3} and \eqref{eqn:CBproof_4}, if $(i,1)\in \mathcal{S}_i(n)$ and $(i',1)\in \mathcal{S}_{i'}(n)$ for all $n\geq 1$, then
\begin{eqnarray}
&&\mathbb{P}\left\{\mu_{i1}-\mu_{i'1}\notin (L_{ii'}(n),U_{ii'}(n)) \mbox{ for some $n\geq 1$}\right\}\nonumber\\
&\leq & \mathbb{P}\left\{\mu_{i1}-\mu_{i'1}\notin \left(\bar{X}_{i1}(n)-\bar{X}_{i'1}(n)-\frac{g_c(t_{i1,i'1}(n))}{t_{i1,i'1}(n)}, \bar{X}_{i1}(n)-\bar{X}_{i'1}(n)+\frac{g_c(t_{i1,i'1}(n))}{t_{i1,i'1}(n)}\right)  \mbox{ for some $n\geq 1$}\right\}\nonumber\\
&=& \mathbb{P} \left\{ t_{i1,i'1}(n)|\bar{X}_{i1}(n)-\bar{X}_{i'1}(n) - (\mu_{i1}-\mu_{i'1}) | \geq g_c(t_{i1,i'1}(n)) \mbox{ for some $n\geq 1$}\right\} \nonumber\\
&=& \mathbb{P}\left\{N_{i1,i'1}<\infty\right\},\label{eq:CI}
\end{eqnarray}
where $N_{i1,i'1}=\min\{n\geq 1: t_{i1,i'1}(n)|\bar{X}_{i1}(n)-\bar{X}_{i'1}(n) - (\mu_{i1}-\mu_{i'1}) | \geq g_c(t_{i1,i'1}(n)) \}$. It follows from Lemma \ref{lem:FirstExitProb}(i) that, letting $\bar{Y}_{ii'}(n)=\bar{X}_{i1}(n)-\bar{X}_{i'1}(n) - (\mu_{i1}-\mu_{i'1})$,
\[\limsup_{\beta\to0} \frac{1}{\beta} \mathbb{P}\left\{t_{i1,i'1}(n)\bar{Y}_{ii'}(n)\leq -g_c(t_{i1,i'1}(n)),\; N_{i1,i'1}<\infty\right\}\leq 1.\]
By the symmetry of the random walk paths,
\[\limsup_{\beta\to0} \frac{1}{\beta} \mathbb{P}\left\{t_{i1,i'1}(n)\bar{Y}_{ii'}(n) \geq g_c(t_{i1,i'1}(n)),\; N_{i1,i'1}<\infty\right\}\leq 1.\]
Hence, in the light of \eqref{eq:CI},
\begin{eqnarray*}
&&\limsup_{\beta\to0} \frac{1}{\beta}\mathbb{P}\left\{\mu_{i1}-\mu_{i'1}\notin (L_{ii'}(n),U_{ii'}(n)) \mbox{ for some $n\geq 1$}\right\}\\
&\leq &\limsup_{\beta\to0} \frac{1}{\beta} \mathbb{P}\left\{N_{i1,i'1}<\infty\right\}\\
&\leq & \limsup_{\beta\to0} \frac{1}{\beta}
\mathbb{P}\left\{t_{i1,i'1}(N_{i1,i'1})\bar{Y}_{ii'}(N_{i1,i'1})\leq -g_c(t_{i1,i'1}(N_{i1,i'1})),\; N_{i1,i'1}<\infty\right\}\\
&&+\limsup_{\beta\to0} \frac{1}{\beta} \mathbb{P}\left\{t_{i1,i'1}(N_{i1,i'1})\bar{Y}_{ii'}(N_{i1,i'1}) \geq g_c(t_{i1,i'1}(N_{i1,i'1})),\; N_{i1,i'1}<\infty\right\} \\
&\leq & 2,
\end{eqnarray*}
which completes the proof.  
\end{proof}

\subsection{Proof of Theorem \ref{theo:sequential}}
In order to prove Theorem \ref{theo:sequential}, we characterize various scenarios that can lead to an incorrect selection (ICS) event that alternative 1 is not ultimately selected. One such scenario is that in step 3.2, i.e., outer-layer elimination, alternative 1 may be eliminated because the approximate dynamic confidence interval for $\mu_{11}-\mu_{i1}$, which is constructed in the spirit of Proposition \ref{prop:robustsequential_CB}, is entirely to the right of the origin. In this scenario, we say ``alternative 1 is \textit{eliminated} by alternative $i$''.

The other possible scenario for ICS is that in step 4, i.e., stopping, the stopping criterion is met with both alternative 1 and alternative $i$ having survived, but alternative 1 has a larger worst-case sample mean than alternative $i$. From now on, when we say ``alternative 1 is \textit{killed} by alternative $i$'', we mean that \textit{either} of the above two scenarios occurs.

\begin{lemma}\label{lem:outer_elim_prob}
Assume that $(1,1)\in \mathcal{S}_1(n)$ and $(i,1)\in\mathcal{S}_i(n)$ for all $n\geq 1$. Then,
\[\limsup_{\beta\to 0}\frac{1}{\beta}\mathbb{P}\{\mbox{alternative 1 is eliminated by alternative $i$}\}\leq 1.\]
\end{lemma}

\begin{proof}[Proof of Lemma \ref{lem:outer_elim_prob}.]
Let $M$ denote the sample size $n$ when the stopping criterion is met. Define
\[
\tilde N_{11,i1} =\min\{n\geq n_0: \tau_{1i}^*(n)[W_{1i}(n)- C_1(n)]\geq g_c(\tau_{1i}^*(n))\mbox{ or } \tau_{1i}^*(n)[W_{1i}(n)+ C_i(n)]\leq -g_c(\tau_{1i}^*(n))\},
\]
where $\tau_{1i}^*(n)=\min_{(1,j)\in \mathcal{S}_1(n),(i,j')\in\mathcal{S}_{i}(n)}\tau_{1j,ij'}(n),
W_{1i}(n)=\max\limits_{(1,j)\in \mathcal{S}_1(n)}\bar{X}_{1j}(n)-\max\limits_{(i,j)\in \mathcal{S}_i(n)}\bar{X}_{ij}(n)$, and $C_i(n) = \max_{(i,j),(i,j')\in \mathcal{S}_i(n)} g_c(\tau_{ij,ij'}(n))/\tau_{ij,ij'}(n)$.
Then, alternative 1 is eliminated by alternative $i$ if and only if
\begin{eqnarray}
&&\{\tau_{1i}^*(\tilde N_{11,i1})[W_{1i}(\tilde N_{11,i1})- C_1(\tilde N_{11,i1})]\geq g_c(\tau_{1i}^*(\tilde N_{11,i1})),\;\tilde N_{11,i1}\leq M<\infty\} \nonumber\\
&\subseteq & \{\tau_{1i}^*(\tilde N_{11,i1})[W_{1i}(\tilde N_{11,i1})- C_1(\tilde N_{11,i1})]\geq g_c(\tau_{1i}^*(\tilde N_{11,i1})),\;\tilde N_{11,i1}<\infty\}. \label{eqn:robustsequential_lemma1_proof1}
\end{eqnarray}

Since $(1,1)\in \mathcal{S}_1(n)$, following the argument for deriving \eqref{eqn:CBproof_2} we can show  that $\max_{(1,j)\in \mathcal{S}_1(n)}\bar{X}_{1j}(n)-\bar{X}_{11} < C_1(n)$. Hence,
\begin{equation}\label{eq:outer_elim_prob_1}
W_{1i}(n)-C_1(n) = \max\limits_{(1,j)\in \mathcal{S}_1(n)}\bar{X}_{1j}(n)-\max\limits_{(i,j)\in \mathcal{S}_i(n)}\bar{X}_{ij}(n)-C_1(n)< \bar{X}_{11}(n) - \bar{X}_{i1}(n)
\end{equation}
It is easy to see that $g_c(t)/t$ is decreasing in $t>0$, and thus
\begin{equation}\label{eq:outer_elim_prob_2}
\frac{g_c{\tau_{1i}^*(n)}}{\tau_{1i}^*(n)} = \max_{(1,j)\in \mathcal{S}_1(n),(i,j')\in\mathcal{S}_{i}(n)}\frac{g_c(\tau_{1j,ij'}(n))}{\tau_{1j,ij'}(n)}
\geq \frac{g_c(\tau_{11,i1}(n))}{\tau_{11,i1}(n)},
\end{equation}
where the inequality holds because $(1,1)\in \mathcal{S}_1(n)$ and $(i,1)\in \mathcal{S}_i(n)$. It follows from \eqref{eqn:robustsequential_lemma1_proof1}, \eqref{eq:outer_elim_prob_1} and \eqref{eq:outer_elim_prob_2} that
\begin{equation*}
0 \leq [W_{1i}(n)-C_1(n)] - \frac{g_c(\tau_{1i}^*(n))}{\tau_{1i}^*(n)} < [\bar{X}_{11}(n) - \bar{X}_{i1}(n)] -\frac{g_c(\tau_{11,i1}(n))}{\tau_{11,i1}(n)}.
\end{equation*}
Therefore,
\begin{eqnarray}
&& \eqref{eqn:robustsequential_lemma1_proof1} \nonumber\\ 
&\subseteq&\left\{\tau_{11,i1}(\tilde N_{11,i1})[\bar{X}_{11}(\tilde N_{11,i1}) - \bar{X}_{i1}(\tilde N_{11,i1})]\geq g_c(\tau_{11,i1}(\tilde N_{11,i1})),\;\tilde N_{11,i1}<\infty\right\} \nonumber \\
&\subseteq&\left\{\tau_{11,i1}(\tilde N_{11,i1})[\bar{X}_{11}(\tilde N_{11,i1}) - \bar{X}_{i1}(\tilde N_{11,i1})-(\mu_{11}-\mu_{i1})]\geq g_c(\tau_{11,i1}(\tilde N_{11,i1})),\;\tilde N_{11,i1}<\infty\right\}. \qquad\qquad \label{eqn:robustsequential_lemma1_proof5}
\end{eqnarray}
Define $N_{11,i1}=\min\{n\geq n_0: \big|\tau_{11,i1}(n)[\bar{X}_{11}(n) - \bar{X}_{i1}(n)-(\mu_{11}-\mu_{i1})]\big| \geq g_c(\tau_{11,i1}(n))\}$.
By \eqref{eqn:robustsequential_lemma1_proof1} and \eqref{eqn:robustsequential_lemma1_proof5}, in order to prove Lemma \ref{lem:outer_elim_prob} it suffices to show
\begin{equation}\label{eq:outer_elim_prob_obj}
\limsup_{\beta\to 0}\frac{1}{\beta}\mathbb{P}\{\tau_{11,i1}(\tilde N_{11,i1})[\bar{X}_{11}(\tilde N_{11,i1}) - \bar{X}_{i1}(\tilde N_{11,i1})-(\mu_{11}-\mu_{i1})]\geq g_c(\tilde N_{11,i1})),\;\tilde N_{11,i1}<\infty\}\leq 1.
\end{equation}

Following the proof of Theorem 2 of \cite{izfree2016} and the functional central limit theorem, it can be shown that the left-hand-side of the inequality \eqref{eq:outer_elim_prob_obj} is upper bounded by its counterpart for the standard Brownian motion, i.e.,
\begin{eqnarray}
&&\limsup_{\beta\to 0}\frac{1}{\beta}\mathbb{P}\{\tau_{11,i1}(\tilde N_{11,i1})[\bar{X}_{11}(\tilde N_{11,i1}) - \bar{X}_{i1}(\tilde N_{11,i1})-(\mu_{11}-\mu_{i1})]\geq g_c(\tilde N_{11,i1})),\;\tilde N_{11,i1}<\infty\} \nonumber\\
&\leq & \limsup_{\beta\to 0}\mathbb{P}\{B(\tilde T_{11,i1})\geq g_c(\tilde T_{11,i1}), \;\tilde T_{11,i1} <\infty\}, \label{eq:dc}
\end{eqnarray}
where $\tilde T_{11,i1}$ is the random time that can be seen as the limit of $\tilde N_{11,i1}$ as $n\to\infty$. Its explicit form can be written by applying the functional central limit theorem, but we omit it since it is quite involved. Moreover, \eqref{eqn:robustsequential_lemma1_proof5} implies that $T\leq \tilde T_{11,i1}$ and $|B(\tilde T_{11,i1})|\geq g_c(\tilde T_{11,i1})$, where $T=\inf\{t> 0: |B(t)|\geq g_c(t)\}$. By the symmetry of standard Brownian motion $B(\cdot)$, we have
\begin{equation}\label{eq:applying_lemma2}
\mathbb{P}\{B(\tilde T_{11,i1})\geq g_c(\tilde T_{11,i1}), \;\tilde T_{11,i1} <\infty\} \leq
\mathbb{P}\{B( T)\geq g_c( T), \; T <\infty\}.
\end{equation}
Combining \eqref{eq:dc} and \eqref{eq:applying_lemma2},
\begin{eqnarray*}
&&\limsup_{\beta\to 0}\frac{1}{\beta}\mathbb{P}\{\tau_{11,i1}(\tilde N_{11,i1})[\bar{X}_{11}(\tilde N_{11,i1}) - \bar{X}_{i1}(\tilde N_{11,i1})-(\mu_{11}-\mu_{i1})]\geq g_c(\tilde N_{11,i1})),\;\tilde N_{11,i1}<\infty\} \\
&\leq &\limsup_{\beta\to 0}\frac{1}{\beta} \mathbb{P}\{B( T)\geq g_c( T), \; T <\infty\} \\
&=&\limsup_{\beta\to 0}\frac{1}{\beta} \cdot \frac{1}{2}e^{-c/2} = 1,
\end{eqnarray*}
where the equality follows from Example 6 in \cite{jennen1981first}. Therefore, \eqref{eq:outer_elim_prob_obj} is true and the proof is complete. 
\end{proof}

\begin{lemma}\label{lem:kill_prob}
Assume that $(1,1)\in \mathcal{S}_1(n)$ and $(i,1)\in\mathcal{S}_i(n)$ for all $n\geq 1$. If $\mu_{i1}-\mu_{11}\geq \delta$, then
\[\limsup_{\beta\to 0}\frac{1}{\beta}\mathbb{P}\{\mbox{alternative 1 is killed by alternative $i$}\}\leq 1.\]
\end{lemma}

\begin{proof}[Proof of Lemma \ref{lem:kill_prob}.]
We follow the notation in the proof of Lemma \ref{lem:outer_elim_prob}. Notice that alternative 1 is killed by alternative $i$ either immediately after the stopping criterion is met, i.e.,
\begin{equation}\label{eq:kill_event1}
\{W_{1i}(M)>0,\;M<\tilde N_{11,i1}\wedge \infty\}\bigcap \left\{\tau_{1i}^*(M)[\delta-C_1(M)\vee C_i(M)]\geq g_c(\tau_{1i}^*(M)) \right\},
\end{equation}
or before the stopping criterion is met, i.e.,
\begin{equation}\label{eq:kill_event2}
\{\tau_{1i}^*(\tilde N_{11,i1})[W_{1i}(\tilde N_{11,i1})- C_1(\tilde N_{11,i1})]\geq g_c(\tau_{1i}^*(\tilde N_{11,i1})),\;\tilde N_{11,i1}\leq M<\infty\}.
\end{equation}
Let $W_{1i}^{0}(n)=W_{1i}(n)-(\mu_{11}-\mu_{i1})$ and
\[
\tilde N_{11,i1}^0 =\min\{n\geq n_0: \tau_{1i}^*(n)[W_{1i}^0(n)- C_1(n)]\geq g_c(\tau_{1i}^*(n))\mbox{ or } \tau_{1i}^*(n)[W_{1i}^0(n)+ C_i(n)]\leq -g_c(\tau_{1i}^*(n))\}.
\]
Then,
\begin{eqnarray}
\eqref{eq:kill_event1} &= &
\{W_{1i}^0(M)>\mu_{i1}-\mu_{11},\;M<\tilde N_{11,i1}\wedge \infty\}\bigcap \left\{\tau_{1i}^*(M)[\delta-C_1(M)\vee C_i(M)]\geq g_c(\tau_{1i}^*(M)) \right\}\nonumber\\
&\subseteq& \{W_{1i}^0(M)> \delta,\; M<\tilde N_{11,i1}\wedge \infty\}\bigcap \left\{\tau_{1i}^*(M)[\delta-C_1(M)]\geq g_c(\tau_{1i}^*(M)) \right\}\nonumber\\
&\subseteq& \left\{\tau_{1i}^*(M)[W_{1i}^0(M)-C_1(M)]\geq g_c(\tau_{1i}^*(M)),\; M<\tilde N_{11,i1}\wedge\infty \right\}\nonumber \\
&\subseteq& \left\{\tau_{1i}^*(M)[W_{1i}^0(M)-C_1(M)]\geq g_c(\tau_{1i}^*(M)),\; \tilde N_{11,i1}^0\leq M<\tilde N_{11,i1}\wedge\infty \right\},\label{eqn:robustsequential_lemma1_proof3}
\end{eqnarray}
where the last step follows from the definition of $\tilde N_{11,i1}^0$. Moreover, notice that
\begin{eqnarray*}
\{M < \tilde N_{11,i1}\} &\subseteq & \{\tau_{1i}^*(n)[W_{1i}(n)+ C_i(n)]> -g_c(\tau_{1i}^*(n))\mbox{ for all $n\leq M$}\} \\
&\subseteq & \{\tau_{1i}^*(n)[W^0_{1i}(n)+ C_i(n)]> -g_c(\tau_{1i}^*(n))\mbox{ for all $n\leq M$}\},
\end{eqnarray*}
since $W^0_{1i}(n)>W_{1i}(n)$. It then follows from \eqref{eqn:robustsequential_lemma1_proof3} that
\begin{eqnarray}
&& \eqref{eq:kill_event1} \nonumber \\
&\subseteq&
\left\{\tau_{1i}^*(N_{11,i1}^0)[W_{1i}^0(N_{11,i1}^0)-C_1(N_{11,i1}^0)]\geq g_c(\tau_{1i}^*(N_{11,i1}^0)), \;\tilde N_{11,i1}^0\leq M<\tilde N_{11,i1}\wedge\infty \right\}.\qquad \qquad \label{eqn:robustsequential_lemma1_proof4}
\end{eqnarray}

For \eqref{eq:kill_event2}, the other scenario that can lead to ICS, we notice that since $W^0_{1i}(n)>W_{1i}(n)$,
\begin{eqnarray}
&& \eqref{eq:kill_event2} \nonumber \\
&\subseteq & \{\tau_{1i}^*(\tilde N_{11,i1})[W^0_{1i}(\tilde N_{11,i1})- C_1(\tilde N_{11,i1})]\geq g_c(\tau_{1i}^*(\tilde N_{11,i1})),\; \tilde N_{11,i1}\leq M<\infty\}\nonumber\\
&\subseteq & \{\tau_{1i}^*(\tilde N_{11,i1})[W^0_{1i}(\tilde N_{11,i1})- C_1(\tilde N_{11,i1})]\geq g_c(\tau_{1i}^*(\tilde N_{11,i1})), \;\tilde N_{11,i1}^0\leq \tilde N_{11,i1}\leq M<\infty\}. \qquad\qquad \label{eq:subseteq_1}
\end{eqnarray}
Moreover,
\begin{eqnarray*}
\{\tilde N_{11,i1}^0\leq \tilde N_{11,i1}\}
&\subseteq & \{\tau_{1i}^*(n)[W_{1i}(n)+ C_i(n)]> -g_c(\tau_{1i}^*(n)) \mbox{ for all $n< \tilde N_{1i}^0$} \}\nonumber\\
&\subseteq & \{\tau_{1i}^*(n)[W^0_{1i}(n)+ C_i(n)]> -g_c(\tau_{1i}^*(n)) \mbox{ for all $n< \tilde N_{1i}^0$} \}.
\end{eqnarray*}
It then follows from \eqref{eq:subseteq_1} that
\begin{eqnarray}
&&\eqref{eq:kill_event2} \nonumber \\
&\subseteq&
 \{\tau_{1i}^*(\tilde N^0_{11,i1})[W^0_{1i}(\tilde N^0_{11,i1})- C_1(\tilde N^0_{11,i1})]\geq g_c(\tau_{1i}^*(\tilde N^0_{11,i1})), \;\tilde N_{11,i1}^0\leq \tilde N_{11,i1}\leq M<\infty\}. \qquad\qquad \label{eq:subseteq_2}
\end{eqnarray}

By \eqref{eqn:robustsequential_lemma1_proof4} and \eqref{eq:subseteq_2},
\begin{eqnarray*}
&&\mathbb{P}\{\mbox{alternative 1 is killed by alternative $i$}\} \\
&\leq &\mathbb{P}\{\tau_{1i}^*(\tilde N^0_{11,i1})[W^0_{1i}(\tilde N^0_{11,i1})- C_1(\tilde N^0_{11,i1})]\geq g_c(\tau_{1i}^*(\tilde N^0_{11,i1})), \;\tilde N_{11,i1}^0 <\infty\}.
\end{eqnarray*}
Hence, in order to prove Lemma \ref{lem:kill_prob}, it suffices to show
\[
\limsup_{\beta\to 0}\frac{1}{\beta}\mathbb{P}\{\tau_{1i}^*(\tilde N^0_{11,i1})[W^0_{1i}(\tilde N^0_{11,i1})- C_1(\tilde N^0_{11,i1})]\geq g_c(\tau_{1i}^*(\tilde N^0_{11,i1})), \;\tilde N_{11,i1}^0 <\infty\}\leq 1.
\]
This can be done by adopting a proof that is essentially identical to the discussion between \eqref{eqn:robustsequential_lemma1_proof1} and the end of the proof of Lemma \ref{lem:outer_elim_prob}.
\end{proof}

To establish Theorem \ref{theo:sequential}, we need one additional building block. Notice that a common assumption shared by both Lemma \ref{lem:outer_elim_prob} and Lemma \ref{lem:kill_prob} is that when alternative $i$ is compared with another alternative, system $(i,1)$, the worst system of alternative $i$, is not yet eliminated in the inner-layer elimination. This assumption essentially guarantees that the worst-case mean performance of alternative $i$ can be accurately estimated via a dynamic confidence interval; see Proposition \ref{prop:robustsequential_CB}. Therefore, we need to characterize the probability that system $(i,1)$ is eliminated by some other system $(i,j)$.

\begin{lemma}\label{lem:robustsequential_inner}
In the inner-layer selection process of the sequential RSB procedure,
\begin{eqnarray*}
\limsup_{\beta\to 0}\ \frac{1}{\beta}\mathbb{P}\{\mbox{system $(i,1)$ is eliminated by system $(i,j)$}\}\leq 1,
\end{eqnarray*}
for each $i=1,2,\ldots,k$ and each $j = 2,3,\ldots,m$.
\end{lemma}
\begin{proof}[Proof of Lemma \ref{lem:robustsequential_inner}.]
Define $N'_{i1,ij}=\min\{n\geq n_0: \tau_{i1,ij}(n)|\bar{X}_{i1}(n)-\bar{X}_{ij}(n)|\geq g_c(\tau_{i1,ij}(n))\}$, for each $i=1,2,\ldots,k$ and each $j = 2,3,\ldots,m$. Then,
\begin{eqnarray*}
&&\lefteqn{\mathbb{P}\{\mbox{system $(i,1)$ is eliminated by system $(i,j)$}\}}\\
& = & \mathbb{P}\left\{\bar{X}_{i1}(N_{i1,ij})-\bar{X}_{ij}(N'_{i1,ij})\leq -\frac{g_c(\tau_{i1,ij}(N'_{i1,ij}))}{\tau_{i1,ij}(N'_{i1,ij})}, N'_{i1,ij}<\infty\right\} \\
&\leq & \mathbb{P}\left\{\bar{X}_{i1}(N_{i1,ij})-\bar{X}_{ij}(N'_{i1,ij})-(\mu_{i1}-\mu_{ij})\leq -\frac{g_c(\tau_{i1,ij}(N'_{i1,ij}))}{\tau_{i1,ij}(N'_{i1,ij})}, N'_{i1,ij}<\infty\right\},
\end{eqnarray*}
where the inequality holds because $\mu_{i1}-\mu_{ij}>0$. The proof is completed by applying Lemma \ref{lem:FirstExitProb}(ii). 
\end{proof}

\begin{proof}[Proof of Theorem \ref{theo:sequential}.]
We first notice that if $\mu_{k1}-\mu_{11}\leq\delta$, then by Assumption \ref{asp:basic} $\mu_{i1}-\mu_{11}\leq \delta$ for all $i=1,\ldots,k$, which implies that $\mathbb{P}(\mu_{i^*1}-\mu_{11}\leq \delta) = 1$ and the theorem trivially holds. Hence, without loss of generality we assume that there exists $l=1,\ldots,k-1$ for which $\mu_{l1}-\mu_{11}\leq \delta$ and $\mu_{l+1,1}-\mu_{11}>\delta$. Then, a good selection event (i.e., alternative $i$ is selected for any $i=1,\ldots,l$) occurs if $\bigcap_{i=l+1}^k\{\mbox{alternative 1 kills alternative $i$}\}.$ We denote
\begin{eqnarray*}
A &=& \bigcap_{i=l+1}^k\{\mbox{alternative 1 kills alternative $i$}\}\\
B &=&\bigcap_{i=2}^l \{\mbox{alternative 1 is not eliminated by alternative $i$}\} \\
C &=& \bigcap_{i=1} ^k \bigcap_{j=2}^m \{\mbox{system $(i,1)$ is not eliminated by system $(i,j)$}\}.
\end{eqnarray*}
Clearly, $\mathbb{P}(A\cap B| C) \geq 1-\mathbb{P}(A^\mathsf{c}|C) - \mathbb{P}(B^\mathsf{c}|C)$. Multiplying $\mathbb{P}(C)$ on both sides this inequality yields
\begin{eqnarray*}
\mathbb{P}(A \cap B \cap C) \geq \mathbb{P}(C) - \mathbb{P}(A^\mathsf{c}\cap C) - \mathbb{P}(B^\mathsf{c}\cap C) = 1- \mathbb{P}(C^\mathsf{c}) - \mathbb{P}(A^\mathsf{c}\cap C) - \mathbb{P}(B^\mathsf{c}\cap C).
\end{eqnarray*}
Since $\mathbb{P}\{\mu_{i^*1}-\mu_{11}\leq  \delta\} \geq \mathbb{P}(A)\geq \mathbb{P}(A \cap B \cap C)$, it follows that
\begin{equation}\label{eq:seq_prob_decompo}
\mathbb{P}\{\mu_{i^*1}-\mu_{11}> \delta\} \leq  \mathbb{P}(A^\mathsf{c}\cap C) +\mathbb{P}(B^\mathsf{c}\cap C)+\mathbb{P}(C^\mathsf{c}).
\end{equation}
Notice that
\begin{eqnarray*}
\mathbb{P}(A^\mathsf{c}\cap C) 
&\leq & \sum_{i=l+1}^k\mathbb{P}\{\mbox{alternative $1$ is killed by alternative $i$, $(1,1)\in\mathcal{S}_1(n)$ and $(i,1)\in\mathcal{S}_i(n)$ for all n}\} \\ 
\mathbb{P}(B^\mathsf{c}\cap C)
&\leq & \sum_{i=1}^l\mathbb{P}\{\mbox{alternative $1$ is eliminated by alternative $i$, $(1,1)\in\mathcal{S}_1(n)$ and $(i,1)\in\mathcal{S}_i(n)$ for all n}\} \\
\mathbb{P}(C^\mathsf{c}) 
&\leq & \sum_{i=1}^k\sum_{j=2}^m \mathbb{P}\{\mbox{system $(i,1)$ is eliminated by system $(i,j)$}\}.
\end{eqnarray*}
Combining Lemma \ref{lem:outer_elim_prob}, Lemma \ref{lem:kill_prob}, Lemma \ref{lem:robustsequential_inner}, and \eqref{eq:seq_prob_decompo}, we have
\begin{eqnarray*}
\limsup_{\alpha\to 0}\frac{1}{\alpha}\mathbb{P}\{\mu_{i^*1}-\mu_{11}> \delta\}
= \limsup_{\beta\to 0}\frac{\beta}{\alpha} \cdot\frac{1}{\beta}\mathbb{P}\{\mu_{i^*1}-\mu_{11}>\delta\}
\leq \frac{1}{km-1} [(k-l) + l + k(m-1)] =1, 
\end{eqnarray*}
which completes the proof.  
\end{proof}

\section{Additional Comparison Between Procedure T and Procedure S}\label{EC_sec:TS}
We fix $k=10$ and plot in Figure \ref{fig:sequential_twostage} the ratio of the average sample size of Procedure T to that of Procedure S (i.e., $N^T/N^S$) as a function of $m$. The result associated with the case of fixing $m=10$ and varying $k$ is almost identical and thus it is omitted.
\begin{figure}[t]
\begin{center}
\caption{Average Sample Sizes of Procedure T and Procedure S Under the EV Configuration}
\includegraphics[width=0.8\textwidth]{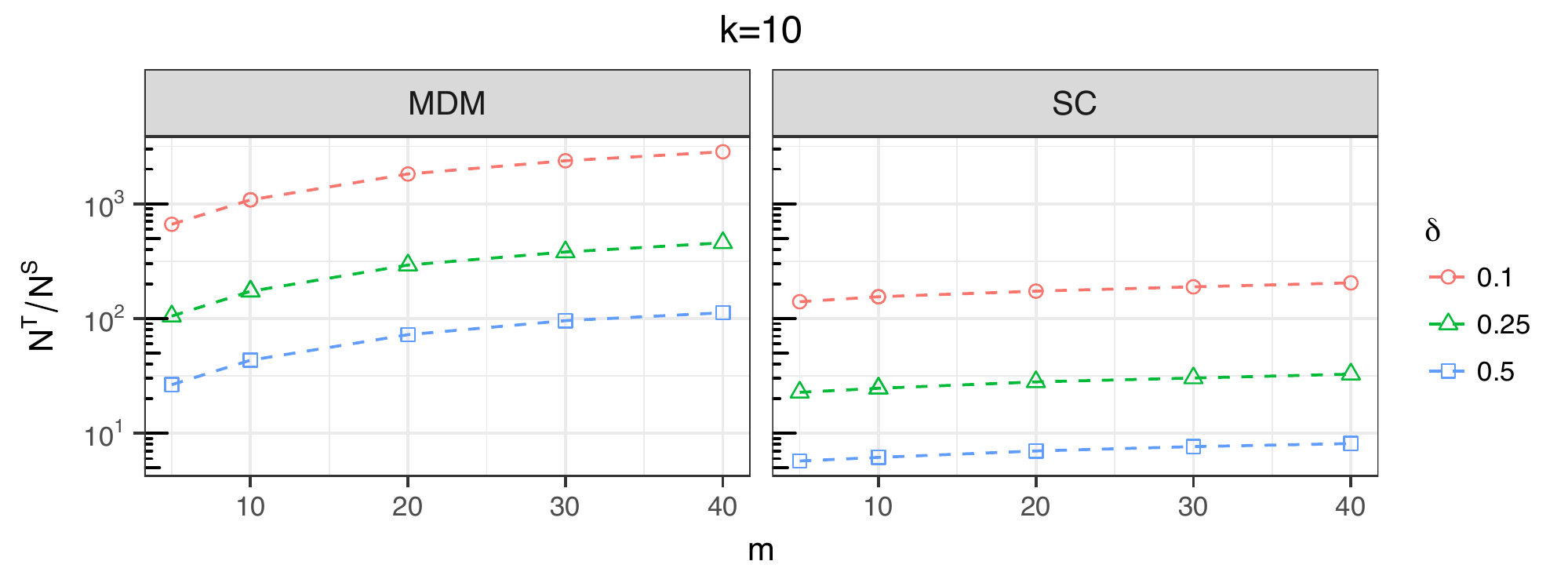}
\end{center}
\small{\textit{Note.} The vertical axis is on a logarithmic scale with base 10.}
\end{figure}

\section{Comparison Between Procedure S and Procedure V}\label{EC_sec:SV}

Procedure 3 in \cite{izfree2016} requires an IZ parameter and we set it to be $\delta/2$ when the procedure is applied to both the inner-layer and outer-layer selection of Procedure V. This is inspired by the decomposition of the IZ parameter in Section \ref{sec:robust_IZ_twostage} for the two-stage RSB procedure.

\vspace{5pt}

\begin{procedure}[Procedure V]
\begin{enumerate}
\item[]
\setcounter{enumi}{-1}
\item \textit{Setup.}
Specify the error allowance $\beta=\alpha/(km-1)$ and the first-stage sample size $n_0\geq 2$. Set $c=-2\log(2\beta)$ and $g_c(t) = \sqrt{[c+\log(t+1)](t+1)}$.
\item \textit{Initialization:} Set $n=n_0$. Set $\mathcal S=\{1,2,\ldots,k\}$ to be  the set of surviving alternatives. Set $\mathcal{S}_i=\{(i,j): j=1,2,\ldots,m\}$ to be  the set of surviving systems of alternative $i$, $i=1,\ldots,k$. Take $n$ independent replications $X_{ij,1},\ldots, X_{ij,n}$ of each system $(i,j)$. Solve $T\delta/2-g_c(T)=0$ for  $T^*$.

\item \textit{Inner-layer Elimination.} For each $i\in \mathcal{S}$, do the following.

\begin{enumerate}[label*=\arabic*]

\item \label{item:update_inner} \textit{Updating.} Compute 
\[
\begin{array}{c}
\displaystyle\bar{X}_{ij}(n)=\frac{1}{n}\sum_{r=1}^{n} X_{ij,r}, \quad i\in\mathcal{S}, \; (i,j)\in \mathcal{S}_i,\\
\displaystyle S_{ij,ij'}^2(n) = \frac{1}{n-1}\sum_{r=1}^{n}\left[X_{ij,r}-X_{ij',r}-(\bar{X}_{ij}(n)-\bar{X}_{ij'}(n))\right]^2,  \quad (i,j),(i,j')\in \mathcal{S}_i.
\end{array}\]

\item \textit{Screening.} 
Compute
\[
\tau_{ij,ij'}(n)=\frac{n}{S^2_{ij,ij'}(n)} \quad\mbox{ and }\quad  Z_{ij,ij'}(n) = \tau_{ij,ij'}(n)[\bar{X}_{ij}(n)-\bar{X}_{ij'}(n)],\quad (i,j),(i,j')\in \mathcal{S}_i.
\]
Assign
$\mathcal{S}_i \gets \mathcal{S}_i\setminus\{(i,j)\in \mathcal{S}_i: Z_{ij,ij'}(n)\leq -g_c(\tau_{ij,ij'}(n)) \mbox{ for some } (i,j')\in \mathcal{S}_i\}$.

\item \textit{Stopping.} If either $|\mathcal{S}_i|=1$ or $\tau_{ij,ij'}(n)\geq T^*$ for all $(i,j),(i,j')\in\mathcal{S}_i$ with $j\neq j'$, then stop and select $j_i^*= \argmax\limits_{j:(i,j)\in \mathcal{S}_i}\bar{X}_{ij}(n)$ as the worst system. Otherwise, take one additional replication of each $(i,j)\in\mathcal{S}_i$ with $i  \in \mathcal{S}$, assign $n\gets n+1$, and return to step \ref{item:update_inner}.

\end{enumerate}

\item \textit{Outer-layer Elimination.} 
\begin{enumerate}[label*=\arabic*]

\item \label{item:update_outer} \textit{Updating.} Compute
\[
\begin{array}{c}
\displaystyle\bar{X}_{ij_i^*}(n)=\frac{1}{n}\sum_{r=1}^{n} X_{ij,r}, \quad i\in\mathcal{S},\\
\displaystyle S_{ij^*_i,i'j^*_{i'}}^2(n) = \frac{1}{n-1}\sum_{r=1}^{n}\left[X_{ij^*_i,r}-X_{i'j^*_{i'},r}-(\bar{X}_{ij^*_i}(n)-\bar{X}_{i'j^*_{i'}}(n))\right]^2,  \quad i,i'\in\mathcal{S}.
\end{array}\]

\item \textit{Screening.} 
For each $i,i'\in \mathcal{S}$ with $i\neq i'$, compute
\[
\tau_{ij^*_i,i'j^*_{i'}}(n)=\frac{n}{S^2_{ij^*_i,i'j^*_{i'}}(n)} \quad\mbox{ and }\quad  Z_{ij^*_i,i'j^*_{i'}}(n) = \tau_{ij^*_i,i'j^*_{i'}}(n)[\bar{X}_{ij^*_i}(n)-\bar{X}_{i'j^*_{i'}}(n)],\quad i,i'\in \mathcal{S}.
\]
Assign
$\mathcal{S} \gets \mathcal{S}\setminus\{(i,j)\in \mathcal{S}: Z_{ij^*_i,i'j^*_{i'}}(n)\geq g_c(\tau_{ij^*_i,i'j^*_{i'}}(n)) \mbox{ for some } i'\in \mathcal{S}\}.$

\item \textit{Stopping.} If either $|\mathcal{S}|=1$ or $\tau_{ij^*_i,i'j^*_{i'}}(n)\geq T^*$ for all $i,i'\in\mathcal{S}$ with $i\neq i'$, then stop and select  $i^*=\argmin\limits_{i\in\mathcal{S}} \bar{X}_{ij^*_i}(n)$
as the best alternative. Otherwise, take one additional replication of each $(i,j^*_i)$ with $i  \in \mathcal{S}$, assign $n\gets n+1$, and return to step \ref{item:update_outer}.\hfill$\Box$
\end{enumerate}

\end{enumerate}
\end{procedure}

The numerical results for the EV configuration are presented in Figure \ref{fig:sequential}. The results for the other two configurations of the variances are very similar so we omit them. 

\begin{figure}[t]
\begin{center}
\caption{Average Sample Sizes of Procedure S and Procedure V Under the EV configuration}\label{fig:sequential}
\includegraphics[width=0.8\textwidth]{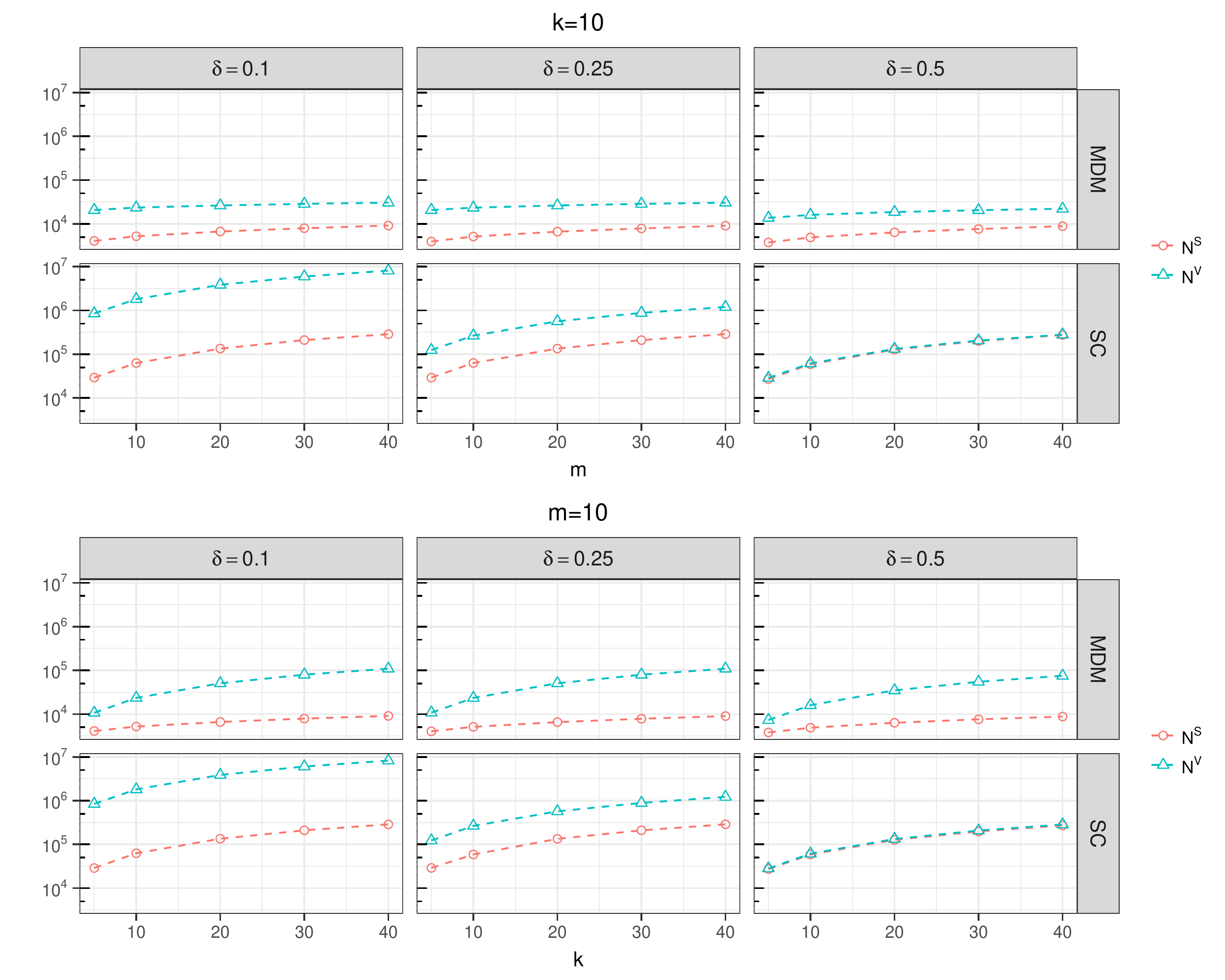}
\end{center}
\small{\textit{Note.} Top: $m$ varies with $k=10$; Bottom: $k$ varies with $m=10$.  The vertical axis is on a logarithmic scale.}
\end{figure}

First, as expected, Procedure S requires significantly fewer samples than Procedure V in general. In particular, under SC,  if the IZ parameter $\delta$ happens to be the difference between the best and the second-best worst-case mean performances (i.e., $\delta = \mu_{21}-\mu_{11}=0.5$), then the average sample sizes required by the two procedures are almost the same, regardless of the problem scale. This implies that in this case, simultaneous elimination of the surviving systems of an alternative that is unlikely to be the best rarely happens in Procedure S, which diminishes its advantage over Procedure V. This is because under SC, the outer-layer selection process deals with the worst-case mean performances $(0,0.5,\dots,0.5)$. With $\delta=0.5$, the alternatives are hard to differentiate in early iterations of Procedure S when the sample size is large enough.

Second, the average sample size of Procedure V is more heavily affected by the configurations of the means than Procedure S. With everything else the same, the average sample size required by Procedure V (denoted by $N^\mathrm{V}$) increases faster than that of Procedure S (denoted by $N^\mathrm{S}$) under MDM than under SC. This suggests that there are a significantly larger number of early outer-layer eliminations in Procedure S under MDM than under SC.

Third, the average sample size of Procedure V is much more sensitive to $\delta$ under SC than that of Procedure S. For instance, with $k=30$ and $m=10$, $N^\mathrm{V}$ increases from about $8.82\times 10^5$ to $5.96\times 10^6$ as $\delta$ drops from 0.25 to 0.1, respectively, whereas $N^\mathrm{S}$ almost remains the value $2.10\times 10^5$. The reason is as follows. The inner-layer selection process of Procedure V that relies on Procedure 3 in \cite{izfree2016} faces systems with equal means under SC. It does not terminate until its stopping criterion that depends on the IZ parameter $\delta$ is met. This stopping criterion is harder to meet for a smaller value of $\delta$. Hence, a smaller $\delta$ implies that the inner-layer selection process of Procedure V needs more time to terminate, resulting in more required samples. 

Last,  $N^\mathrm{V}$  grows much faster than $N^\mathrm{S}$ as the problem scale $k$ or $m$ increases. For instance, with $\delta=0.25$, $k=10$ and MDM, $N^\mathrm{V}$ increases from about $2.90\times 10^4$ to $2.87\times 10^5$ as $m$ increases from 5 to 40, whereas $N^\mathrm{S}$ increases from about $3.94\times 10^3$ to $9.04\times 10^3$. This is because, as the problem scale increases, there are more opportunities for Procedure S to eliminate alternatives early, leading to a slower growth in $N^{\mathrm{S}}$.

The above numerical comparison between  Procedure S and Procedure V indicates that the inferior performance of the latter stems from its non-fully sequential nature -- its outer-layer selection cannot begin unless all the inner-layer eliminations are completed. Hence, an alternative having a configuration of the means that is close to the  SC will dominate the inner-elimination time, even if it were otherwise a poor alternative for outer elimination. This suggests an additional comparison between Procedure S and Procedure V using a configuration of the means that somewhat combines MDM and SC as follows 
\begin{equation*}\label{eqn:mdm_sc_mixed}
[\mu_{ij}]_{k\times m} =
\begin{pmatrix}
0 && -0.2 && -0.2 && \ldots && -0.2 \\
0.5 && 0.3 && 0.3 && \ldots && 0.3 \\
\vdots && \vdots && \vdots && \ddots && \vdots \\
0.5(k-1) && 0.5(k-1)-0.2 && 0.5(k-1)-0.2 && \ldots && 0.5(k-1)-0.2
\end{pmatrix}.
\end{equation*}
Here, the alternatives are ordered as MDM, but the systems of each alternative are ordered as SC. 

We adopt the EV configuration of the variances. The other experiment specifications remain the same. The results are given in Figure \ref{fig:additional_VS}
and they are consistent with the findings revealed by Figure \ref{fig:sequential}. 
\begin{figure}[t]
\begin{center}
\caption{Average Sample Sizes of Procedure S and Procedure V}\label{fig:additional_VS}
\includegraphics[width=0.8\textwidth]{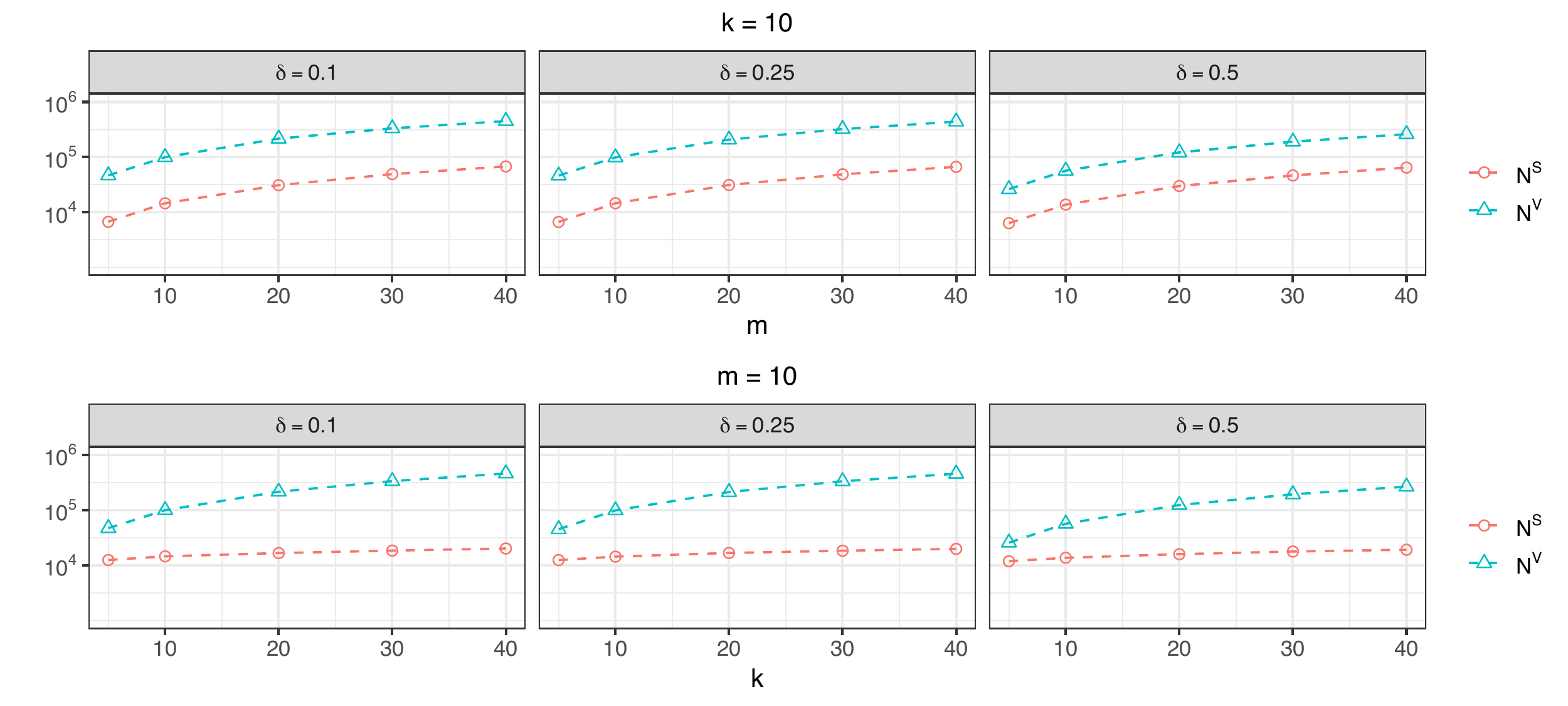}
\end{center}
\small{\textit{Note. }Top: $m$ varies with  $k=10$; Bottom: $k$ varies with $m=10$. The vertical axis is on a logarithmic scale.}
\end{figure}

\section{Realized PCS of the $G/G/s + G$ Queueing Example}\label{EC_sec:queue}
This section assesses efficacy of the two proposed RSB procedures as to whether they can achieve the target PCS as promised. This complements the analysis in Section \ref{sec:robust_numerical}, since the samples in queueing simulation are not normally distributed in general, in contrast to the normal assumptions of the  numerical experiments there.

With $\sigma=2$ and $\ell =50$, we run 1,000 macro-replications of the experiment below. 
\begin{enumerate}[label=(\roman*)]
\item 
Generate a sample of service times from $P_0$.
\item 
Construct an ambiguity set $\mathcal{P}$ based on the sample.
 \item 
Compute the expected cost $\mathbb{E}[f(s, \xi)]$ with 10,000 samples for each pair $(s, P)$, $s=1,\ldots,k$, $P\in\mathcal{P}$ so that the estimation errors are negligible and find the best alternative, 
 \item  
Run the two RSB procedures 1,000 times independently on $\mathcal P$ and estimate their respective PCS. The IZ parameter $\delta$ is set to be small enough so that the indifference zone contains only the best alternative. 
\end{enumerate}
In summary, there are 100 ambiguity sets constructed in total, each from one macro-replication. Hence, PCS is estimated 100 times for each RSB procedure and some statistics of these estimated probabilities are reported in Table \ref{tab:pcs}. Clearly, both procedures can achieve the target PCS even by a large margin in general, despite the samples' non-normal distribution. Moreover, the two-stage RSB procedure is significantly more conservative than the sequential RSB procedure, producing a larger realized PCS. This reflects that the former requires a larger number of samples, which is consistent with the findings in Section \ref{sec:robust_numerical}.

\begin{table}[ht]
\begin{center}
\caption{Realized PCS}\label{tab:pcs}
\begin{tabular}{cccccccccccc}
\toprule
\multirow{2}{*}{Procedure} && \multicolumn{9}{c}{Statistics}                         \\
\cmidrule{3-11}
                           && Min   && 25\% Quantile && Median && 75\% Quantile && Max   \\
\midrule
Two-stage                  && $0.992$ && $0.998$         && $0.999$  && $1.000$         && $1.000$ \\
Sequential                 && $0.951$ && $0.980$         && $0.991$  && $0.996$         && $1.000$ \\
\bottomrule
\end{tabular}
\end{center}
\small{\textit{Note.} Target PCS: 0.95.}
\end{table}

% Acknowledgments here
\section*{Acknowledgments}
\small{The authors would like to thank the associate editor and three anonymous referees for their insightful and invaluable comments that have significantly improved this paper. The preliminary work of this paper \citep{FanHongZhang13} was presented at the 2013 Winter Simulation Conference. The first author was supported by Natural Science Foundation of China (Grant 71701196). The second author was supported by Hong Kong Research Grants Council (GRF 16203214) and Natural Science Foundation of China (Grant 71720107003). The third author was supported by Hong Kong Research Grants Council (TRS No. T32-102/14N and GRF 16211417).}

\bibliographystyle{chicago}
\bibliography{robustselection}
\end{document}